\documentclass[%
 reprint,
superscriptaddress,
 amsmath,amssymb,
 aps,
]{revtex4-2}

\usepackage{graphicx}% Include figure files
\usepackage{dcolumn}% Align table columns on decimal point
\usepackage{amsthm}
\usepackage{amsmath,amssymb,amsfonts,times}
\usepackage{graphicx}
\usepackage{textcomp}
\usepackage{xcolor}
\usepackage{subfigure}
\usepackage{epsfig}
\usepackage{multirow}
\usepackage{algorithm}
\usepackage{algorithmicx}
\usepackage{algpseudocode}
\usepackage{array}
\usepackage{epstopdf}
\usepackage{color}
\usepackage{diagbox}
\usepackage{bm}
\usepackage{bbm}
\usepackage{mathrsfs}
\usepackage{tcolorbox}
\usepackage{pdfsync}

\usepackage{url}
\usepackage[hidelinks]{hyperref}
\usepackage{paralist}
\usepackage{xcolor}
\usepackage{color}
\usepackage{comment}
\usepackage{multirow}
\usepackage[toc, page]{appendix}

\usepackage{ulem}

\usepackage{fancyhdr}
\usepackage{cleveref}
\newtheorem{lemma}{Lemma}%[section] %%    with section number.
\newtheorem{cor}{Corollary}%[section]
\newtheorem{theorem}{Theorem}

\newcommand{\R}{\mathbb{R}}
\def \real    { \mathbb{R} }

\newcommand{\C}{\mathbb{C}}

\newcommand{\e}{\begin{equation}}
\newcommand{\ee}{\end{equation}}
\newcommand{\en}{\begin{equation*}}
\newcommand{\een}{\end{equation*}}
\newcommand{\eqn}{\begin{eqnarray}}
\newcommand{\eeqn}{\end{eqnarray}}
\newcommand{\bmat}{\begin{bmatrix}}
\newcommand{\emat}{\end{bmatrix}}
\DeclareMathAlphabet\mathbfcal{OMS}{cmsy}{b}{n}
\renewcommand{\P}[1]{\operatorname{\mathbb{P}}\left(#1\right)}
\newcommand{\E}{\operatorname{\mathbb{E}}}

\newcommand{\vct}[1]{\boldsymbol{#1}}
\newcommand{\mtx}[1]{\boldsymbol{#1}}
\newcommand{\<}{\langle}
\renewcommand{\>}{\rangle}
\newcommand{\trace}{\operatorname{trace}}

\newcommand{\set}[1]{\mathbb{#1}}
\DeclareMathOperator*{\argmin}{\text{arg~min}}

\def \st {\operatorname*{s.t.\ }}
\newcommand{\wh}{\widehat}

\newcommand{\wt}{\widetilde}
\newcommand{\nqbit}{n}

\newcommand{\norm}[2]{\left\| #1 \right\|_{#2}}

\newcommand{\bracket}[1]{\left( #1 \right)}

\newcommand{\innerprod}[2]{\left\langle #1,  #2 \right\rangle}

\newcommand{\calN}{\mathcal{N}}

\newcommand{\calP}{\mathcal{P}}

\newcommand{\ve}{\vct{e}}

\newcommand{\vg}{\vct{g}}

\newcommand{\vu}{\vct{u}}

\newcommand{\vphi}{\vct{\phi}}

\newcommand{\vrho}{\vct{\rho}}
\newcommand{\vzero}{\vct{0}}

\newcommand{\mA}{\mtx{A}}
\newcommand{\mB}{\mtx{B}}

\newcommand{\mD}{\mtx{D}}

\newcommand{\mF}{\mtx{F}}

\newcommand{\mU}{\mtx{U}}
\newcommand{\mV}{\mtx{V}}

\newcommand{\mX}{\mtx{X}}

\newcommand{\mSigma}{\mtx{\Sigma}}

\newcommand{\mId}{{\bf I}}

\newcommand{\setX}{\set{X}}

\graphicspath{{./figs/}}

\newlength{\imgwidth}
\definecolor{orange}{RGB}{255,127,50}

\newcommand{\imag}{\mathrm{i}}

\newboolean{twoColVersion}
\setboolean{twoColVersion}{false}
\newcommand{\twoCol}[2]{\ifthenelse{\boolean{twoColVersion}} {#1} {#2} }

\begin{document}

\preprint{APS/123-QED}

\title{Enhancing Quantum State Reconstruction with  Structured Classical Shadows}% Force line breaks with \\

\author{Zhen Qin}
\email{qin.660@osu.edu}
\affiliation{Department of Computer Science and Engineering, The Ohio
State University, Columbus, Ohio 43210, USA}
\author{Joseph M. Lukens}
\email{jlukens@purdue.edu}
\affiliation{Elmore Family School of Electrical and Computer Engineering and Purdue Quantum Science and Engineering Institute, Purdue University, West Lafayette, Indiana 47907, USA}
\affiliation{Quantum Information Science Section, Oak Ridge National Laboratory, Oak Ridge, Tennessee 37831, USA}
\affiliation{Research Technology Office and Quantum Collaborative, Arizona State University, Tempe, Arizona 85287, USA}
\author{Brian T. Kirby}
\email{brian.t.kirby4.civ@army.mil}
\affiliation{DEVCOM Army Research Laboratory, Adelphi, MD 20783, USA}
\affiliation{Tulane University, New Orleans, LA 70118, USA}
\author{Zhihui Zhu}
\email{zhu.3440@osu.edu}
\affiliation{Department of Computer Science and Engineering, The Ohio
State University, Columbus, Ohio 43210, USA}

\begin{abstract}
Quantum state tomography (QST) remains the prevailing method for benchmarking and verifying quantum devices; however, its application to large quantum systems is rendered impractical due to the exponential growth in both the required number of total state copies and classical computational resources. Recently, the classical shadow (CS) method has been introduced as a more computationally efficient alternative, capable of accurately predicting key quantum state properties. Despite its advantages, a critical question remains as to whether the CS method can be extended to perform QST with guaranteed performance.
In this paper, we address this challenge by introducing a projected classical shadow (PCS) method with guaranteed performance for QST based on Haar-random projective measurements. PCS extends the standard CS method by incorporating a projection step onto the target subspace. For a general quantum state consisting of $n$ qubits, our method requires a minimum of $O(4^n)$ total state copies to achieve a bounded recovery error in the Frobenius norm between the reconstructed and true density matrices, reducing to $O(2^n r)$ for states of rank $r<2^n$---meeting information-theoretic optimal bounds in both cases.
For matrix product operator states, we demonstrate that the PCS method can recover the ground-truth state with $O(n^2)$ total state copies, improving upon the previously established Haar-random bound of $O(n^3)$. Simulation results further validate the effectiveness of the proposed PCS method.
\end{abstract}

\maketitle

\section{Introduction}

Quantum state tomography (QST) is widely used for estimating quantum states~\cite{bertrand1987tomographic,vogel1989determination,leonhardt1995quantum,hradil1997quantum,James2001}. To reconstruct the density matrix with high accuracy, measurements should be performed on a large number of identical copies; specifically, for single-copy (i.e., non-collective)
measurements,
a minimum of $O(4^n)$ total copies is required to estimate the density matrix of an $n$-qubit system with a bounded recovery error, as defined by the Frobenius norm between the reconstructed and true density matrices \cite{haah2017sample}. Various methods have been proposed to achieve efficient and accurate QST. Classical computational approaches include linear inversion \cite{fano1957description}, maximum likelihood estimation \cite{hradil1997quantum,vrehavcek2001iterative,James2001}, Bayesian inference \cite{blume2010optimal,granade2016practical,lukens2020practical}, region estimation \cite{blume2012robust,faist2016practical}, classical machine learning \cite{lohani2020machine}, and least squares estimators \cite{kyrillidis2018provable,brandao2020fast,zhu2024connection}. In contrast, quantum machine learning methods encompass algorithms such as variational quantum circuits \cite{sen2022variational,liu2020variational}, quantum principal component analysis \cite{lloyd2014quantum}, and quantum variational algorithms combined with classical statistics \cite{kurmapu2020machine}.

\begin{figure*}[!bt] %[!thbp]
   %\hspace{-0.4cm}%
   \centering
   \includegraphics[width=12cm, keepaspectratio]%
   {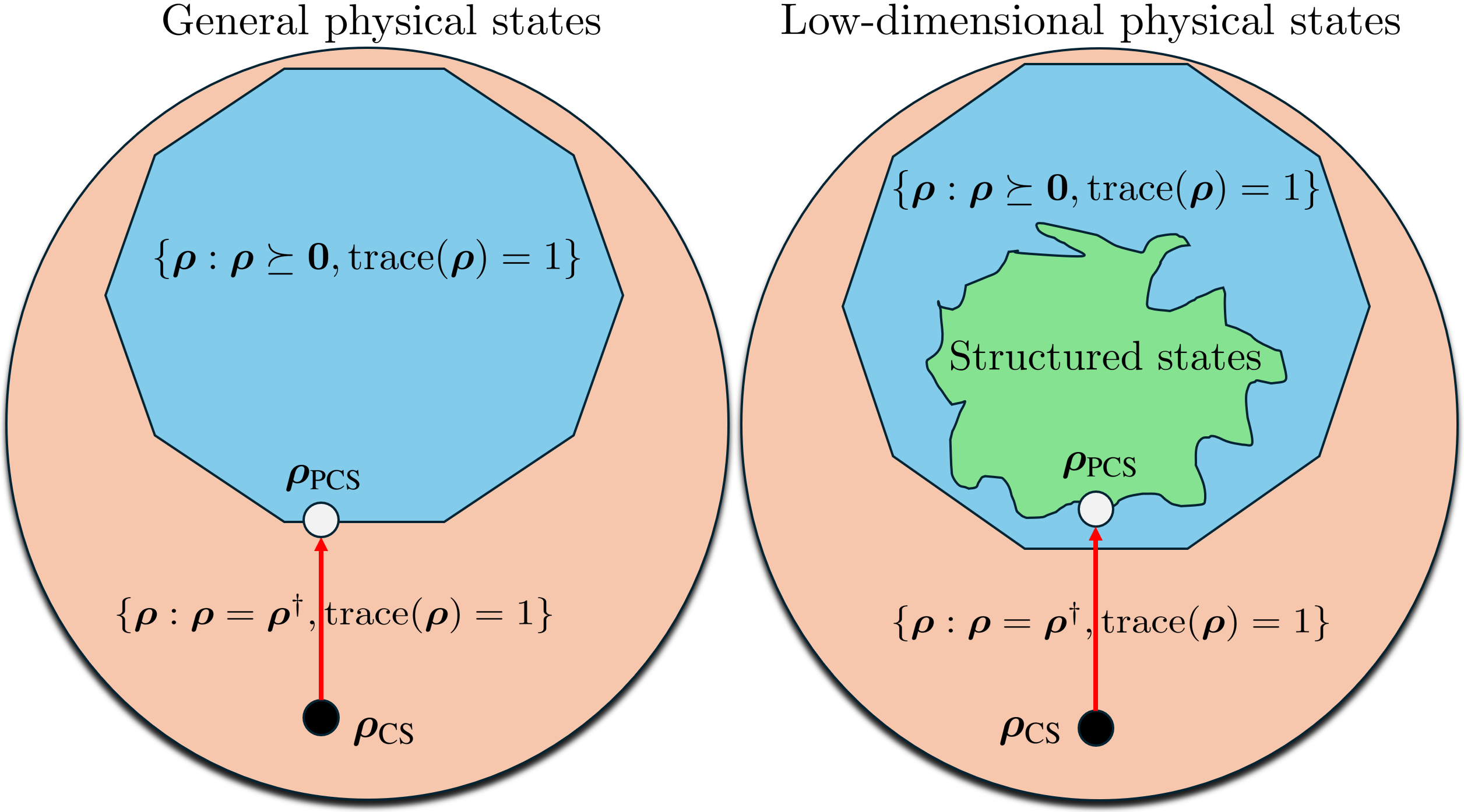}
 \vspace{0in}
   \caption{Illustration of proposed PCS method. Given an initial CS estimate $\vrho_{\text{CS}}$ lying in the space of Hermitian and unit-trace matrices (not necessarily PSD), we compute the closest state $\vrho_{\text{PCS}}$ in the physical space of interest---either the space of all possible states (left) or a subspace possessing a desired structure (right).}
   \label{Illustration of projected classical shadow method}
\end{figure*}

A significant reduction in the number of required state copies can be achieved by assuming two common low-dimensional structures: low-rankness and matrix product operators (MPOs). $(i)$ Low-rank density matrices frequently emerge in quantum systems with pure or nearly pure states that exhibit low entropy~\cite{KuengACHA17,guctua2020fast,francca2021fast,voroninski2013quantum,haah2017sample},
and low-rank assumptions are employed in various state estimation procedures, with a range of associated measurement processes, including 4-designs \cite{KuengACHA17}, Pauli strings~\cite{liu2011universal,guctua2020fast}, Clifford gates~\cite{francca2021fast}, and Haar-random projective measurements~\cite{voroninski2013quantum}. When the density matrix has rank $r$, the required number of total state copies can be reduced to $O(2^n r)$ \cite{haah2017sample,francca2021fast}, yet this remains exponential in $n$, posing challenges for current quantum computers exceeding 100 qubits. $(ii)$ MPOs, on the other hand, offer a more scalable alternative for certain quantum systems, including one-dimensional spatial systems \cite{eisert2010}, Hamiltonians with decaying long-range interactions \cite{pirvu2010matrix}, and states generated by noisy quantum devices \cite{noh2020efficient}. When employing Haar-random projective measurements \cite{qin2024quantum} or specific classes of informationally complete positive operator-valued measures (IC-POVMs) \cite{qin2024sample}, the required number of total state copies can be reduced to polynomial scaling---either $O(n^3)$ or $O(n)$, respectively---while ensuring bounded recovery error for MPO states.

While algorithms with low-rank assumptions or low-dimensional structures can enable significantly improved scaling, they still face considerable computational complexity, which in existing approaches can be attributed to four potential operations: $(i)$ the calculation of the inverse; $(ii)$ repeated inner product operations between matrices that grow exponentially with $n$; $(iii)$ multiple projection steps onto the target subspace; and $(iv)$ additional matrix multiplications introduced by nonconvex algorithms to enforce low-rankness or MPO representations.
Recently, an efficient and experimentally feasible approach, known as classical shadow (CS) estimation, was introduced by Ref.~\cite{huang2020predicting} to infer limited sets of state properties like fidelity, entanglement measures, and correlations. By exploiting efficient computational and storage capabilities on classical hardware, all necessary processing to predict these properties can be carried out via classical computations. This has sparked a series of studies leveraging the CS method \cite{acharya2021shadow,struchalin2021experimental,akhtar2023scalable,grier2024sample,ippoliti2024classical,becker2024classical}. However, existing research predominantly focuses on {\it{state properties}}, raising the critical question of whether the CS method can be effectively extended to the {\it{full state}} (i.e., QST) with guaranteed performance.

In this paper, we derive performance guarantees for QST using a method we term projected classical shadow (PCS), which projects CS estimators onto target subspaces of the Hilbert space, as illustrated in  Fig.~\ref{Illustration of projected classical shadow method}.
Given that the original CS density matrix is Hermitian but not in general positive semidefinite (PSD), our method involves projecting its eigenvalues onto the simplex~\cite{chen2011projection}. We demonstrate that this approach requires $O(4^n)$ total state copies to achieve a bounded recovery error in the Frobenius norm. For low-rank states, we further leverage (truncated) low-rank eigenvalue decomposition and show that the required number of total state copies can be reduced to $O(2^n r)$ for the same accuracy. Finally, for MPO states, we employ a quasi-optimal MPO projection---tensor-train singular value decomposition (TT-SVD) \cite{Oseledets11} with a simplex projection---to form the PCS  step, demonstrating that with $O(n^2)$ total state copies, the method reliably recovers the ground-truth state. While suboptimal relative to the degrees of freedom for MPO states, this approach improves upon the theoretical $O(n^3)$ scaling in Ref.~\cite{qin2024quantum}.  PCS also offers a framework for incorporating prior knowledge about the target state form into the CS approach.

\paragraph*{Notation:} We use bold capital letters (e.g., $\bm{X}$) to denote matrices,  bold lowercase letters (e.g., $\bm{x}$) to denote column vectors, and italic letters (e.g., $x$) to denote scalar quantities. Matrix elements are denoted in parentheses. For example, $\mX(i_1, i_2)$ denotes the element in position
$(i_1, i_2)$ of the matrix $\mX$.   The superscripts $(\cdot)^\top$ and $(\cdot)^\dagger$ denote the transpose and Hermitian transpose, respectively. For two matrices $\mA,\mB$ of the same size, $\innerprod{\mA}{\mB} = \trace(\mA^\dagger\mB)$ denotes the inner product.
$\|\mX\|$, $\|\mX\|_1$, and $\|\mX\|_F$ respectively represent the spectral, trace, and Frobenius norm of $\mX$.
For two positive quantities $a,b\in \real^+$, the inequality $b\lesssim a$ or $b = O(a)$ implies $b\leq c a$ for some universal constant $c$; likewise, $b\gtrsim a$ or $b = \Omega(a)$ represents $b\ge ca$ for some universal constant $c$.

\section{Classical Shadows}
\label{sec:classical shadow}

Quantum information science harnesses quantum states for information processing \cite{nielsen2000quantum}. The state of an $n$-qubit system can be  described by the density operator $
\vrho\in\C^{2^n\times 2^n}$, which is PSD ($\vrho \succeq \vzero$) and has unit-trace ($\trace(\vrho) = 1$). In order to estimate this states, measurements can be performed on a collection of copies.

\paragraph*{Projective measurements:}
Within the most general quantum measurement framework of positive operator valued measures (POVMs) \footnote{Specifically, a POVM is characterized as a set of PSD matrices: $\{\mA_1,\ldots,\mA_K \}\in\C^{2^n\times 2^n}, \ \ \st  \sum_{k=1}^K \mA_k = \mId_{2^n}$. Each POVM element $\mA_k$ corresponds to a potential outcome of a quantum measurement with the special case of projective measurements corresponding to the case where all $\mA_k$ are pairwise orthogonal projection operators, meaning they satisfy $\mA_k^{2}=\mA_k$ and $\mA_{k}\mA_{j}=0$ for $k\ne j$},
the special case of  projective measurements is often employed, where the measurement outcomes are associated with an orthonormal eigenbasis of the system.
To implement such a measurement defined by an arbitrary orthonormal basis $\{\vphi_k : \vphi_k^\dagger\vphi_l =\delta_{kl}\}$, we can introduce a unitary matrix $\mU = \begin{bmatrix} \vphi_1 & \cdots & \vphi_{2^n} \end{bmatrix}\in\C^{2^n\times 2^n}$ and apply $\mU^\dagger$ to the state $\vrho$ before conducting a projective measurement in the computational basis $\{\ve_k\}$, where $\mU \ve_k = \vphi_k$. The probability of observing the $k$-th outcome is given by:
\begin{eqnarray}
p_k = \<\vphi_k\vphi_k^\dagger, \vrho  \> = \ve_k^\dagger \bracket{\mU^{\dagger} \vrho \mU} \ve_k.
\label{Probability of rank-one orth POVM}
\end{eqnarray}
However, a single projective measurement,
even if repeated infinitely many times, provides only partial information on $\vrho$, so
multiple projective measurements must be conducted in various bases. In the subsequent discussion, we denote the number of distinct measurement bases by $M$, and the measurement operators for the $m$-th projective measurement by $ \{\vphi_{m,1}\vphi_{m,1}^\dagger,\ldots, \vphi_{m,2^n}\vphi_{m,2^n}^\dagger   \}$.

\paragraph*{Classical shadow (CS):}
Consider the original CS proposal with single-shot Haar-random projective measurements.
Given an unknown $n$-qubit ground truth $\vrho^\star$, we repeatedly execute the measurement procedure above Eq.~\eqref{Probability of rank-one orth POVM} in which $\mU$ is chosen randomly from the Haar distribution and each measurement is performed on only one copy (i.e., a new $\mU$ is selected for each copy measured) .
The specific result $\ve_{j_m}$ yields a snapshot, or ``shadow,'' of the underlying quantum state, which for Haar-distributed unitaries can be expressed as~\cite{huang2020predicting}:
\begin{eqnarray}
    \label{one time classical shadow}
    \vrho_m &=& (2^n + 1)\mU_m\ve_{j_m}\ve_{j_m}^\dagger \mU_m^\dagger   - \mId_{2^n}\nonumber\\
    &=& (2^n+1) \vphi_{m,j_m}\vphi_{m,j_m}^\dagger - \mId_{2^n}.
\end{eqnarray}
By construction, this snapshot equals the ground truth in expectation (over both unitaries and measurement outcomes): $\E[\vrho_m] = \vrho^\star$. Executing this process $M$ times produces an array of $M$ independent classical snapshots for the total CS estimator:
\begin{eqnarray}
    \label{classical shadow}
    \vrho_{\text{CS}} &=& \frac{1}{M}\sum_{m=1}^M\vrho_m\nonumber\\
    &=& \frac{1}{M}\sum_{m=1}^M \left[(2^n+1) \vphi_{m,j_m}\vphi_{m,j_m}^\dagger - \mId_{2^n}\right].
\end{eqnarray}
\paragraph*{CS for Tomography?} Although CS estimators can efficiently predict observables of $\vrho^\star$, to our knowledge there exist no results concerning recovery error of the \textit{full state}. Following the detailed derivation in Appendix~\ref{proof of expectation of MSE}, we find the expectation of the mean squared error:
\begin{eqnarray}
    \label{proof of expectation of MSE main result}
     \E\|\vrho_{\text{CS}}  - \vrho^\star \|_F^2 = \frac{4^n + 2^n - 1 - \|\vrho^\star \|_F^2}{M}.
\end{eqnarray}
Given that $\|\vrho^\star \|_F^2 \leq \left[\trace(\vrho^\star)\right]^2=1$, it follows that Eq.~\eqref{proof of expectation of MSE main result} can be simplified to
\begin{eqnarray}
    \label{proof of expectation of MSE main result 1}
     \E\|\vrho_{\text{CS}}  - \vrho^\star \|_F^2 \approx \frac{4^n}{M}
\end{eqnarray}
for large $n$. Eq.~\eqref{proof of expectation of MSE main result 1} demonstrates that  stable recovery of the full state can be achieved only when $M$ scales proportionally to $4^n$, aligning with the optimal $M$ required in QST for general states~\cite{haah2017sample}.

A comparison between CS and traditional QST returns several key observations of relevance to this study:
\begin{enumerate}
\item CS yields an unbiased estimate ($\E[\vrho_{\text{CS}}] = \vrho^\star$), whereas the solution from QST  is often biased~\cite{schwemmer2015systematic}.
\item While the CS estimator is typically unphysical (not PSD), leading QST methods like MLE~\cite{James2001}, projected least squares~\cite{smolin2012efficient}, and Bayesian inference~\cite{blume2010optimal} enforce physicality by construction.
\item CS boasts significantly lower computational complexity compared to QST.
\item For $M\ll 2^n$, CS outperforms QST in predicting certain linear observables, not in predicting the entire state  \cite{lukens2021bayesian,zhu2024connection}.
\item Including prior information about state structure allows for a reduction in scaling in QST (see \Cref{Methods comparison LR,Methods comparison MPO}). Currently no known method for similarly reducing CS scaling exists. In other words, CS requires $O(4^n)$ measurements for estimating the full state, as demonstrated in Eq. \eqref{proof of expectation of MSE main result 1}.
\end{enumerate}
In the next section we investigate methods for incorporating prior information about state structure into CS to reduce the scaling shown in Eq. \eqref{proof of expectation of MSE main result 1}.

\section{Projected Classical Shadow (PCS) for QST}
\label{sec: Projected classical shadow}
In this section, we will study the application of CS for the task of describing the full quantum state and show that, with a simple projection step, CS estimators are also effective for QST and achieve (nearly) information-theoretically optimal bounds for broad classes of states.
Let $\setX$ denote the class of states of interest, and assume that the underlying ground truth $\vrho^\star \in \setX$. For instance,  $\setX$ could contain all physical states (PSD and unit-trace) or be restricted to a specific structure with compact representations, such as low-rank or MPO states.
We then define $\vrho_{\text{PCS}}$ as the projection of $\vrho_{\text{CS}}$ on the set $\setX$ that minimizes Frobenius error, i.e.,
\begin{eqnarray}
    \label{projected classical shadow any set}
    \vrho_{\text{PCS}} =  \calP_{\setX}(\vrho_{\text{CS}}) := \argmin_{\vrho \in \setX}\norm{\vrho - \vrho_{\text{CS}}}{F}.
\end{eqnarray}

To provide a unified and general analysis of Eq.~\eqref{projected classical shadow any set}, we enlist tools from $\epsilon$-net and covering number theory to capture the complexity of the classes of states within the set $\setX$. First, consider the set $\calN =
\left\{\frac{\vrho}{\|\vrho\|_F} : \vrho\in\setX\right\}$ scaled to unit Frobenius norm.
For $\epsilon>0$, the set $\calN_\epsilon \subset \calN$ is said to be an $\epsilon$-net (or an $\epsilon$-cover) over $\calN$ if for all $\frac{\vrho}{\|\vrho\|_F}\in\calN$, there exists $\frac{\vrho'}{\|\vrho'\|_F}\in\calN_\epsilon$ such that $\left\|\frac{\vrho}{\|\vrho\|_F} - \frac{\vrho'}{\|\vrho'\|_F}\right\|_F\leq \epsilon$.  The size of an $\epsilon$-net with the smallest cardinality  is called the covering number of $\setX$, denoted by $N_\epsilon(\setX)$. Intuitively speaking, a covering number is the minimum number of balls of a specified radius $\epsilon$ to cover a given set entirely. Coverings are useful for managing the complexity of a large set: instead of directly analyzing the behavior of an uncountable number of points in $\calN$, we can analyze the finite number of points in $\calN_\epsilon$. The behavior of all points in $\calN$ is similar to that of the points in $\calN_\epsilon$, as each point in $\calN$ is close to some point in the covering.

Instead of the covering number $N_\epsilon(\setX)$, our analysis will rely on the covering number of the set $\overline\setX$ formed by the differences between the elements in $\setX$:
\begin{eqnarray}
\label{the set of difference of the target set}
\overline\setX = \bigg\{\vrho_1 - \vrho_2: \  \vrho_1,\vrho_2 \in \setX, \vrho_1\neq \vrho_2\bigg\}.
\end{eqnarray}
In many cases, the covering number $N_\epsilon(\overline\setX)$ can be upper bounded by $N_\epsilon^2(\setX)$.
Here we use $\overline \setX$ for convenience in the following.

The covering number when $\setX$ comprises all physical quantum states can be computed as $\log N_{\epsilon}(\overline \setX) = O(4^n\log\frac{9}{\epsilon} )$. By comparison, for quantum states with rank at most $r$, this reduces to $\log N_{\epsilon}(\overline \setX) = O(2^n r\log\frac{9}{\epsilon})$; when the density matrices are represented by MPOs with bond dimension $D$, the covering number can be further reduced to $\log N_{\epsilon}(\overline \setX) = O\left(4 n D^2 \log \frac{4n+\epsilon}{\epsilon}\right)$, as discussed in Sec.~\ref{sec: MPO Projected classical shadow}.

\begin{theorem}
\label{project CS:general} %[Guarantee with exact projection]
For a given $\vrho^\star \in \setX$, let $\vrho_{\textup{PCS}}$ be the projected CS in Eq.~\eqref{projected classical shadow any set}.  Then with probability at least $1- e^{- \Omega(\log N_{1/2}(\overline \setX))}$,
\begin{eqnarray}
    \label{error bound of projected classical shadow MPO F1}
    \|\vrho_{\textup{PCS}} - \vrho^\star\|_F \leq O\left(\sqrt{\frac{\log N_{1/2}(\overline\setX)}{M}}\right).
\end{eqnarray}
\end{theorem}
The proof is given in {Appendix} \ref{proof of error bound for general state}. Here the set $\setX\subset\{\vrho\in\C^{2^n\times 2^n}: \vrho = \vrho^\dagger, \trace(\vrho) =1\}$ is any subspace of Hermitian, trace-one matrices (tan space in  Fig.~\ref{Illustration of projected classical shadow method}. The set $\setX$ will be specialized to PSD matrices only  (blue space in Fig.~\ref{Illustration of projected classical shadow method}) in \Cref{project CS conclusion full rank} and low-dimensional structures (green space in Fig.~\ref{Illustration of projected classical shadow method}) in \Cref{project CS conclusion,project CS MPO conclusion}.
\Cref{project CS:general} guarantees a stable recovery of the ground-truth $\vrho^\star$ with $\xi$-closeness in the Frobenius norm, provided that the number of Haar-random projective measurements $M$ satisfies $M\geq \Omega(\log N_{1/2}(\set \mX)/\xi^2)$, which scales linearly with the logarithm of the covering number. For structured sets $\setX$ that are nonconvex, such as MPO states, computing the optimal projection $\calP_{\setX}$ might be difficult or even NP-hard. For these cases, we can use numerical methods to compute an approximate projection $\wt\calP_{\setX}$ that we assume is $\alpha$-approximately optimal ($\alpha\geq 1$), satisfying
\begin{eqnarray}
\label{expansiveness property of approximate projection}
\wt\calP_{\setX} (\vrho) \in \setX, \quad \norm{\wt\calP_{\setX} (\vrho)  - \vrho}{F} \le \sqrt{\alpha} \norm{ \calP_{\setX} (\vrho) - \vrho}{F}
\end{eqnarray}
for any $\vrho$.
Denote by $\wt\vrho_{\text{PCS}} =  \wt\calP_{\setX}(\vrho_{\text{CS}})$ the PCS estimator obtained with this approximate projection. The following extends the results in \Cref{project CS:general} to $\wt\vrho_{\text{PCS}}$.

\begin{theorem}
\label{project CS:general approximation}
For a given $\vrho^\star \in \setX$, let $\wt\vrho_{\textup{PCS}}$ be the approximate  PCS estimator in Eq.~\eqref{expansiveness property of approximate projection}.  Then with probability at least $1- e^{- \Omega(\log N_{1/2}(\overline\setX))}$,
\begin{eqnarray}
    \label{error bound of projected classical shadow MPO F1 approximation}
    \|\wt\vrho_{\textup{PCS}} - \vrho^\star\|_F \leq O\left(\sqrt{\frac{\alpha\log N_{1/2}(\overline\setX)}{M}}\right).
\end{eqnarray}
\end{theorem}

\subsection*{General physical states}
\label{general physical states}
We first specialize $\setX$ to all physical quantum states~\footnote{We chose the label ``simplex'' for this set since the eigenvalues $\{\lambda_k\}$ of all physical states define a standard simplex, i.e., $\lambda_k\geq 0$ and $\sum_k\lambda_k =1$.}:
\begin{eqnarray}
\label{Set of general quantum state}
\setX_{\text{simplex}} = \{ \vrho\in\C^{2^n\times 2^n}:  \vrho\succeq {\bm 0}, \trace(\vrho) = 1 \}.
\end{eqnarray}
For $\setX_\text{simplex}$, we can achieve the PCS projection in Eq.~\eqref{projected classical shadow any set} by performing an eigenvalue decomposition and projecting the eigenvalues to the simplex after the algorithm of Ref.~\cite{smolin2012efficient}.  Since the corresponding set $\overline\setX$ has covering number $\log N_\epsilon(\overline\setX) = O(4^n\log\frac{9}{\epsilon})$, we can plug this information into \Cref{project CS:general} to obtain recovery guarantee for $\calP_{\text{simplex}}(\vrho_{\text{CS}})$.

\begin{cor}
\label{project CS conclusion full rank}
For a given physical state $\vrho^\star\in\C^{2^n\times 2^n}$, we perform $M$ projective measurements to obtain the CS estimate $\vrho_{\textup{CS}}$. Then with probability at least $1  -  e^{- \Omega(4^n)}$, the projected classical shadow $\calP_{\textup{simplex}}(\vrho_{\textup{CS}})$ satisfies
\begin{eqnarray}
    \label{error bound of projected classical shadow F full rank}
    &&\|\calP_{\textup{simplex}}(\vrho_{\textup{CS}}) - \vrho^\star \|_F \leq O\left( \sqrt{\frac{4^n }{M}}  \right).
\end{eqnarray}
\end{cor}

\subsection*{Low-rank states}
\label{sec: Low-rank Projected classical shadow}
We next explore the structure of pure or nearly pure quantum states characterized by low entropy and represented as low-rank density matrices.
Assuming $\vrho^\star$ has rank $r\leq2^n$, we can refine our attention to the set $\setX_r = \{\vrho\in\C^{2^n\times 2^n}: \vrho\succeq {\bm 0}, \trace(\vrho) = 1, \text{rank}(\vrho) = r\}$.
Denote $\calP_{\setX_r}(\cdot)$
as the optimal projection satisfying Eq.~\eqref{projected classical shadow any set}.  It follows from \Cref{project CS:general} and the covering number of the corresponding set $\log N_{\epsilon}(\overline \setX) = O(2^n r\log\frac{9}{\epsilon})$ that $\|\calP_{\setX_r}(\vrho_{\text{CS}}) - \vrho^\star \|_F\leq O( \sqrt{2^nr/M}  )$.

However, since we are unaware of an algorithm to perform the ideal projection $\calP_{\setX_r}(\cdot)$, we instead consider a two-step alternative to obtain the low-rank projected classical shadow (LR-PCS):
\begin{eqnarray}
    \label{projected classical shadow}
    \vrho_{\text{LR-PCS}} =   \calP_{\text{simplex}}(\calP_{\text{rank-$r$}}(\vrho_{\text{CS}})),
\end{eqnarray}
where $\calP_{\text{rank-$r$}}(\cdot)$ denotes the rank-$r$ projection obtained by setting all eigenvalues beyond the $r$-th largest eigenvalue to zero. We can show that $\vrho_{\text{LR-PCS}}$ shares a similar guarantee as $\calP_{\setX_r}(\vrho_{\text{CS}})$.
\begin{theorem}
\label{project CS conclusion}
Given $M$ Haar-random projective measurements on physical state $\vrho^\star\in\setX_r$,  with probability $1  -  e^{- \Omega(2^n r)}$ $\vrho_{\textup{LR-PCS}}$ defined in Eq.~\eqref{projected classical shadow} satisfies
\begin{eqnarray}
    \label{error bound of projected classical shadow F}
    &&\|\vrho_{\textup{LR-PCS}} - \vrho^\star \|_F \leq O\left( \sqrt{\frac{2^nr }{M}}  \right).
\end{eqnarray}
\end{theorem}
The detailed proof appears in {Appendix} \ref{Proof of error bound in projected CS}. This theoretical recovery error is optimal, given that the degrees of freedom for the ground truth $\vrho^\star$ are $O(2^nr)$. This highlights that LR-PCS can achieve the optimal solution in QST using independent measurements, without requiring multiple iterations of optimization algorithms.

To compare LR-PCS with prior results, we convert the result of \Cref{project CS conclusion} to trace norm leveraging the inequality between the Frobenius and the trace norms~\cite{haah2017sample}, namely $\|\vrho_{\text{LR-PCS}} - \vrho^\star \|_1 \leq \sqrt{2r}\|\vrho_{\text{LR-PCS}} - \vrho^\star \|_F \leq O( \sqrt{2^nr^2/M} )$, which matches the optimal guarantee (up to small log terms) with independent measurements according to \cite{haah2017sample}. We have summarized the comparison in Table~\ref{Methods comparison LR}.
\begin{table}[!htbp]
\renewcommand{\arraystretch}{1.2}
\begin{center}
\caption{Comparing  the total number of state copies in PCS using single-shot global Haar unitaries to that in optimal QST. Here, $n$ denotes the number of qubits, $r$ represents the rank of the target state, and $\xi$ signifies the desired precision in trace distance, i.e., $\|\wh\vrho - \vrho^\star \|_1\leq \xi$ for estimator $\wh\vrho$.}
\label{Methods comparison LR}
{\begin{tabular}{|c||c|c|} \hline {Methods} & {$\vrho^\star\in\setX_\text{simplex}$} & { $\vrho^\star\in\setX_r$}
\\\hline
{PCS} &  $\Omega(8^n /\xi^2)$ &  $\Omega(2^nr^2 /\xi^2)$   \\
\hline
{Optimal QST \cite{haah2017sample}} &   $\Omega(8^n /\xi^2)$  &  $\Omega(2^nr^2 /(\xi^2\log(1/\xi)))$  \\
\hline
\end{tabular}}
\end{center}
\end{table}

\subsection*{MPO states}
\label{sec: MPO Projected classical shadow}
While the computational and storage requirements for low-rank density matrices are significantly smaller compared to general ones, they still grow exponentially in the number of qubits $n$. Moreover, the assumption of high purity on which the low-rank approximation is based becomes increasingly tenuous in practice for existing processors in the noisy intermediate-scale quantum (NISQ) era. For this reason, reducing parameter count through alternative assumptions is worth pursuing.
Examples such as ground states of many quantum systems with short-range interactions and states generated by such systems within a finite duration \cite{eisert2010} often possess entanglement localized to subsystems of the entire quantum computer. Consequently, they can be compactly represented using MPOs, whose degrees of freedom scale only polynomially in $n$. To assist in the development of an MPO-PCS method, we will first establish their connection to tensor train (TT) decompositions \cite{Oseledets11}, a technique widely utilized in signal processing and machine learning.

For a $n$-qubit density matrix $\vrho^\star\in\C^{2^n\times 2^n}$,
we employ a single index array $i_1\cdots i_\nqbit$ ($j_1\cdots j_\nqbit$) to denote the row (column) indices, where $i_1,\ldots,i_\nqbit\in \{1,2\}$ \footnote{ Specifically, $i_1\cdots i_n$ represents the $(i_1+\sum_{\ell=2}^n 2^{\ell-1}(i_\ell-1))$-th row.}.  We designate $\vrho^\star$ as an MPO if we can represent its $(i_1\cdots i_\nqbit,j_1\cdots j_\nqbit)$-th element using the following matrix product~\cite{werner2016positive}:
\begin{eqnarray}
\label{DefOfMPO}
\vrho^\star(i_1 \cdots i_\nqbit, j_1 \cdots j_\nqbit)  =  \mX_1^{i_1,j_1} \mX_2^{i_2,j_2} \cdots \mX_\nqbit^{i_\nqbit,j_\nqbit},
\end{eqnarray}
where $\mX_\ell^{i_\ell,j_\ell}\in \C^{D\times D}$ for $\ell \in\{ 2,\dots,n-1\}$,  $\mX_1^{i_1,j_1}\in \C^{1\times D}$, $\mX_n^{i_n,j_n}\in \C^{D\times 1}$, and $D$ is the {\it bond dimension}, and thus we can introduce the set of physical MPO states with bond dimension $D$ as
\begin{eqnarray}
\label{SetOfMPO}
\setX_{D}= \Big\{ \vrho\in\C^{2^n\times 2^n}:\  \vrho\succeq {\bm 0}, \trace(\vrho) = 1, \atop \text{bond dimension}(\vrho) = D\Big\}.
\end{eqnarray}
Here the corresponding difference set $\overline\setX$ has  covering number $\log N_{\epsilon}(\overline \setX) = O\left(4 n D^2 \log \frac{4n+\epsilon}{\epsilon}\right)$, which is proportional to the degrees of freedom $O(4 n D^2)$ in the MPO states.
Given the optimal projection $\calP_{\setX_{D}}(\cdot)$, it follows from \Cref{project CS:general} that $\|\calP_{\ \setX_{D}}\left(\vrho_{\text{CS}}\right) - \vrho^\star\|_F \leq O\left(\sqrt{n D^2 \log n/M}\right)$.

However, we have been unable to implement the optimal $\calP_{\setX_D}(\cdot)$ due to the difficulty in satisfying both the MPO and simplex conditions simultaneously. Therefore, we introduce a quasi-optimal projection based on a sequential singular value decomposition (SVD) algorithm, commonly referred to as tensor train SVD (TT-SVD) \cite{Oseledets11}. Based on tensor-matrix equivalence, we can design a two-step MPO PCS method:
\begin{eqnarray}
    \label{projected classical shadow in the MPO}
    \vrho_{\text{MPO-PCS}} =  \calP_{\text{simplex}}( \text{SVD}_{D}^{tt}(\vrho_{\text{CS}})),
\end{eqnarray}
where $\text{SVD}_{D}^{tt}(\cdot)$ denotes the TT-SVD operation.
It is worth noting that the bond dimension of $\vrho_{\text{MPO-PCS}}$ may differ slightly from $D$ due to the simplex projection, but the recovery error still depends on $D$. We analyze the recovery error of Eq.~\eqref{projected classical shadow in the MPO} as follows:
\begin{theorem}
\label{project CS MPO conclusion}
Consider an MPO state $\vrho^\star\in\setX_D$, measured $M$ times with Haar-random projections. With $\vrho_{\textup{MPO-PCS}}$ defined as in Eq.~\eqref{projected classical shadow in the MPO}, with probability $1  -  e^{- \Omega(nD^2 \log n)}$ we have
\begin{eqnarray}
    \label{error bound of projected classical shadow MPO F}
    \|\vrho_{\textup{MPO-PCS}} - \vrho^\star\|_F \leq O\left(\sqrt{\frac{n^2D^2 \log n}{M}}\right).
\end{eqnarray}
\end{theorem}

\begin{figure*}[!th]
\centering
\subfigure{
\begin{minipage}[t]{0.31\textwidth}
\centering
\includegraphics[width=5.5cm]{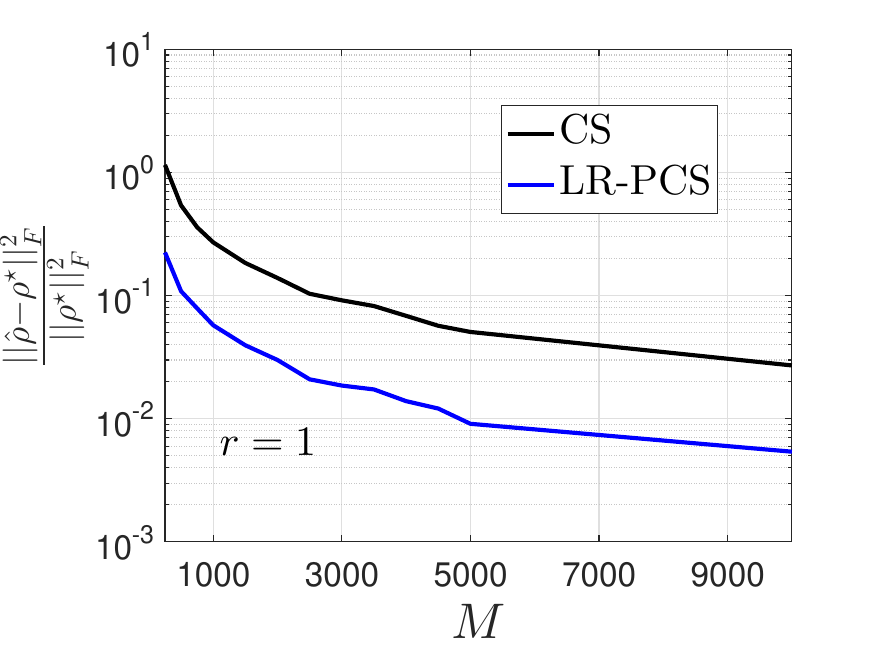}
%\caption{(A)}
\end{minipage}
}
\subfigure{
\begin{minipage}[t]{0.31\textwidth}
\centering
\includegraphics[width=5.5cm]{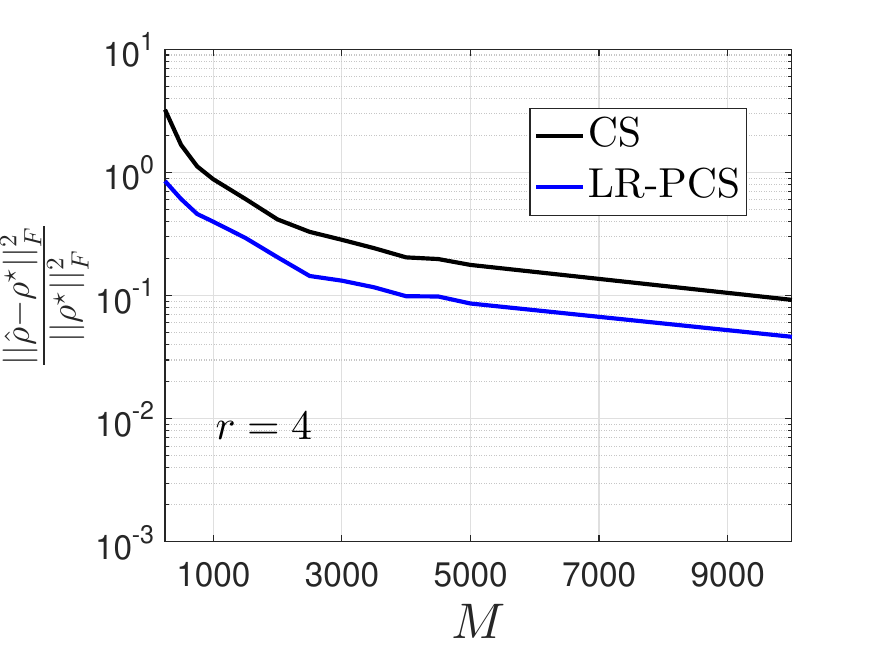}
%\caption{(A)}
\end{minipage}
}
\subfigure{
\begin{minipage}[t]{0.31\textwidth}
\centering
\includegraphics[width=5.5cm]{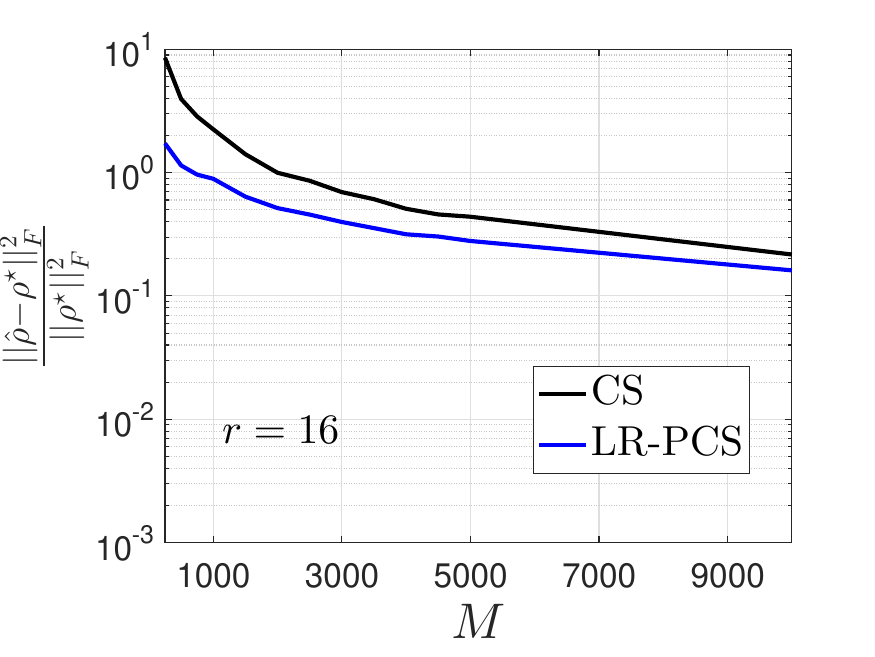}
%\caption{(B)}
\end{minipage}
}
\caption{Mean squared error as a function of state copies $M$ for CS and LR-PCS methods on $n=4$ qubits, averaged over trials on ten randomly chosen ground truth states for each rank $r\in\{1,4,16\}$. The figures span $M=250$ to $M=10000$.}
\label{comparison between diff methods in full density matrix}
\end{figure*}

The proof can be found in {Appendix} \ref{Proof of project CS MPO conclusion}. Note that due to the quasi-optimality of the TT-SVD, the upper bound of Eq.~\eqref{error bound of projected classical shadow MPO F} is not optimal when considering the degrees of freedom $O(4n D^2)$ in $\vrho^\star$. To our knowledge, there exists no method that can guarantee both MPO and PSD constraints simultaneously.
Should such an optimal MPO projection be found, however, we could potentially remove one factor of $n$ in the numerator of Eq.~\eqref{error bound of projected classical shadow MPO F}, thus ensuring exact MPO rank.
In Table~\ref{Methods comparison MPO}, we summarize the total number of state copies required for MPO-PCS compared to existing QST methods. It is important to highlight that the QST results represent sufficient, rather than necessary, conditions. Compared with these results, MPO-PCS still demonstrates favorable performance.

\begin{table}[!htbp]
\renewcommand{\arraystretch}{1.2}
\begin{center}
\caption{Total number of copies in MPO-PCS compared to MPO-based QST using Haar measures, and spherical $3$-designs. Here, $n$ denotes the number of qubits, and $\zeta$ signifies the desired precision in Frobenius distance, i.e., $\|\wh\vrho - \vrho^\star \|_F\leq \zeta$ for estimator $\wh\vrho$.}
\label{Methods comparison MPO}
{\begin{tabular}{|c||c|} \hline {Method} &  { $\vrho\in\setX_D$}
\\\hline
{Approximate PCS (Haar)} &  $\Omega(n^2D^2 \log n /\zeta^2)$   \\
\hline
{Optimal PCS (Haar)} &  $\Omega(nD^2 \log n /\zeta^2)$   \\
\hline
{QST (Haar) \cite{qin2024quantum}} &   $\Omega(n^3D^2 \log n /\zeta^2)$    \\
\hline
{QST (spherical $3$-designs) \cite{qin2024sample} } &   $\Omega(nD^2 \log n /\zeta^2)$    \\
\hline
\end{tabular}}
\end{center}
\end{table}

\section{Simulation Results}
\label{sec: experiments}
In this section,  we conduct numerical QST experiments with Haar-random projective measurements to compare CS, LR-PCS, and MPO-PCS methods.  For each configuration, we conduct 10 Monte Carlo tomographic experiments in which each Haar measurement and result are sampled at random; then we take the average over all 10 trials to report the results. For the random state cases (Figs.~\ref{comparison between diff methods in full density matrix},\ref{MPS T2 MSE diff summary}), each trial corresponds to a different randomly chosen ground truth, whereas for the tailored state cases (Figs.~\ref{thermal state and GHZ state MSE},\ref{F norm for different n}), each trial in a given average is performed on the same ground truth.

In the first set of tests, we compare CS and LR-PCS for a specific rank $r$ as a function of measurements $M$. We generate random ground-truth density matrices $\vrho^\star = \mF^\star {\mF^\star}^\dagger\in\C^{16\times 16}$ ($n=4$ qubits), where $\mF^\star = \frac{ \mA^\star + \imag\mB^\star}{\|\mA^\star + \imag\mB^\star\|_F}\in \C^{16\times r}$, and the entries of $\mA^\star$ and $\mB^\star$ are independent and identically distributed (i.i.d.) samples drawn from the standard normal distribution. Notably, when $r = 16$, LR-PCS reduces to projection onto the set of general physical states defined in Eq.~\eqref{Set of general quantum state}. The results in Fig.~\ref{comparison between diff methods in full density matrix} for rank $r\in\{1,4,16\}$ reveal two key observations: $(i)$ as the rank $r$ decreases and the number of measurements $M$ increases, the recovery error across all methods consistently reduces, with the squared error quantitatively scaling as expected ($4^n/M$ for CS and $2^nr/M$ for LR-PCS); and $(ii)$ for any $r$ and $M$,  LR-PCS outperforms standard CS (even at full rank), as it preserves physicality under any rank constraints.

\begin{figure}[!ht]
\centering
\subfigure{
\begin{minipage}[t]{0.22\textwidth}
\centering
\includegraphics[width=4.5cm]{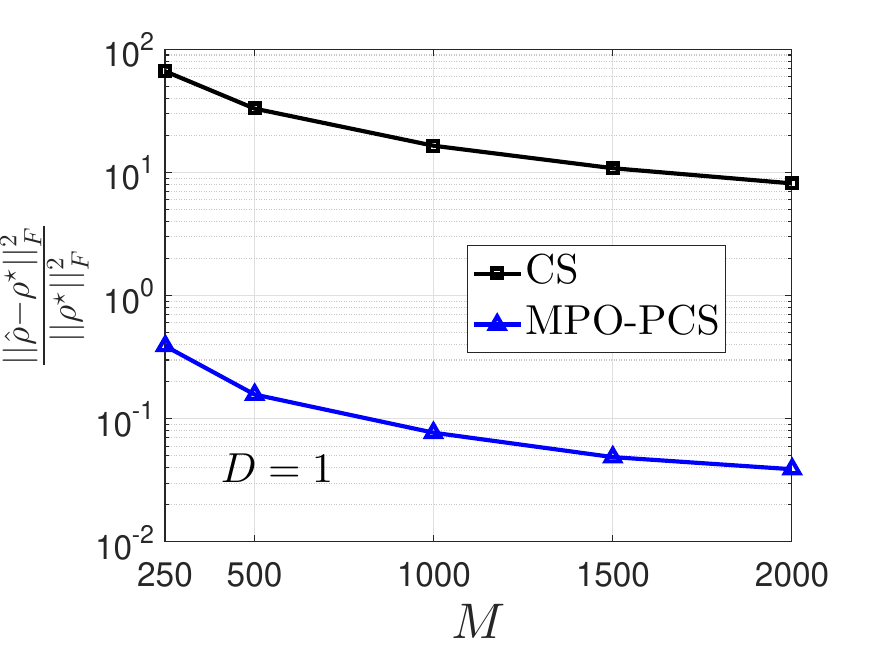}
%\caption{(A)}
\end{minipage}
\label{MPS T2 MSE diff1}
}
\subfigure{
\begin{minipage}[t]{0.22\textwidth}
\centering
\includegraphics[width=4.5cm]{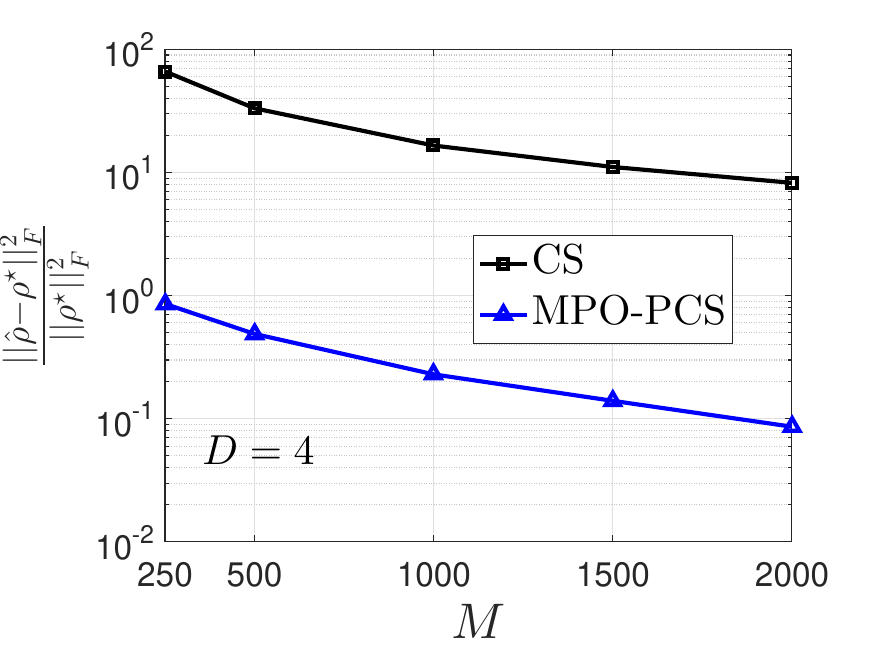}
%\caption{(A)}
\end{minipage}
\label{MPS T2 MSE diff2}
}
\caption{Mean squared error as a function of state copies $M$ for CS and MPO-PCS methods on seven-qubit MPO states, where each point is an average over trials on ten randomly chosen ground truth states for each bond dimension $D\in\{1,4\}$.}
\label{MPS T2 MSE diff summary}
\end{figure}

\begin{figure*}[!ht]
\centering
\subfigure[]{
\begin{minipage}[t]{0.31\textwidth}
\centering
\includegraphics[width=5.5cm]{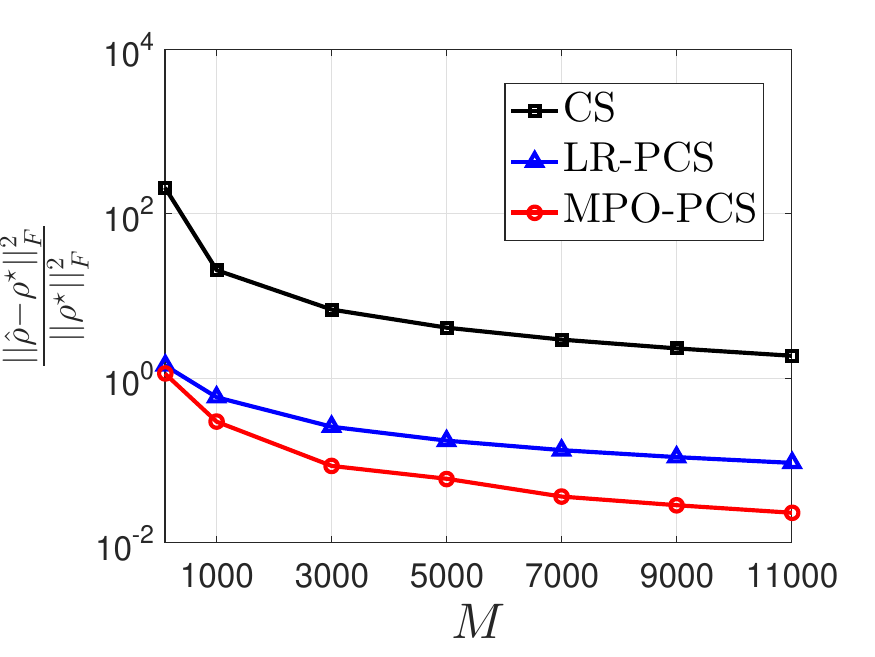}
%\caption{(A)}
\end{minipage}
\label{thermal T0.2 MSE}
}
\subfigure[]{
\begin{minipage}[t]{0.31\textwidth}
\centering
\includegraphics[width=5.5cm]{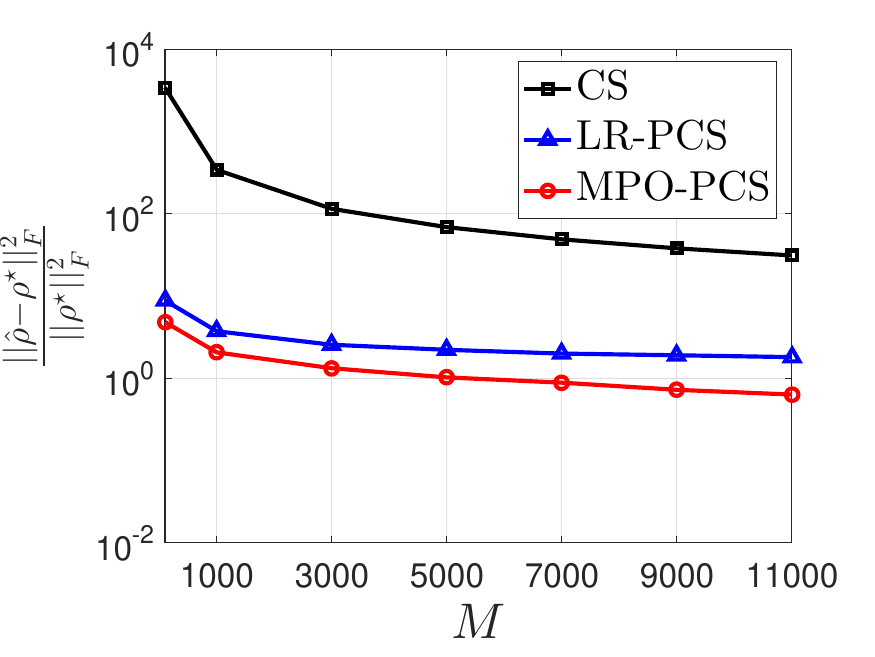}
%\caption{(A)}
\end{minipage}
\label{thermal T2 trace}
}
\subfigure[]{
\begin{minipage}[t]{0.31\textwidth}
\centering
\includegraphics[width=5.5cm]{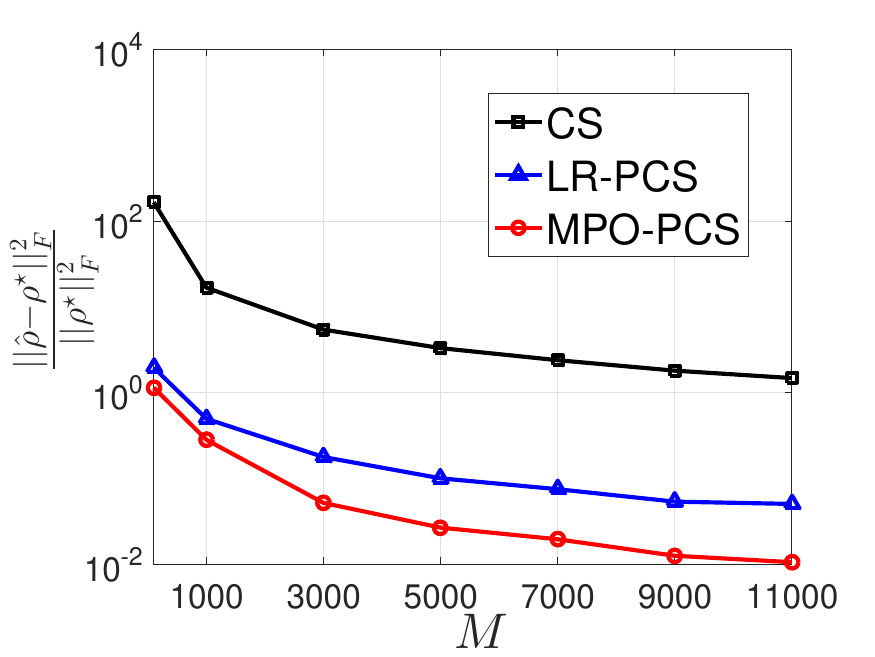}
%\caption{(A)}
\end{minipage}
\label{GHZ trace}
}
\caption{Mean square error as a function of the number of state copies $M$ for (a) thermal state ($T = 0.2$), (b) thermal state ($T = 2$), and (c) GHZ state. Comparison between different methods for (a) thermal state ($T = 0.2$), (b) thermal state ($T = 2$), and (c) GHZ state. All figures have $M=100$ as the starting point.}
\label{thermal state and GHZ state MSE}
\end{figure*}

\begin{figure*}[!ht]
\centering
\subfigure[]{
\begin{minipage}[t]{0.31\textwidth}
\centering
\includegraphics[width=5.5cm]{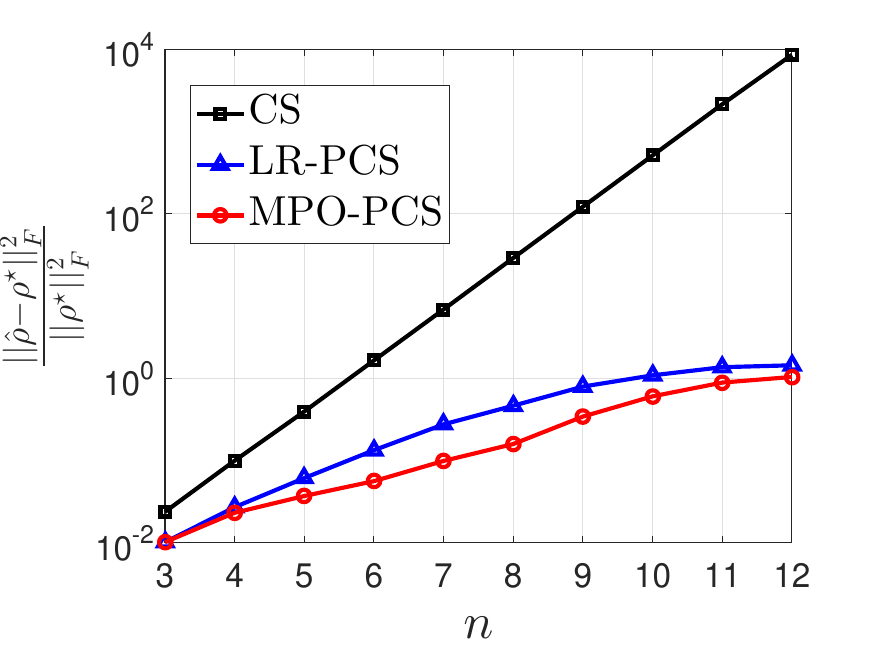}
%\caption{(A)}
\end{minipage}
\label{MSE diff n thermal}
}
\subfigure[]{
\begin{minipage}[t]{0.31\textwidth}
\centering
\includegraphics[width=5.5cm]{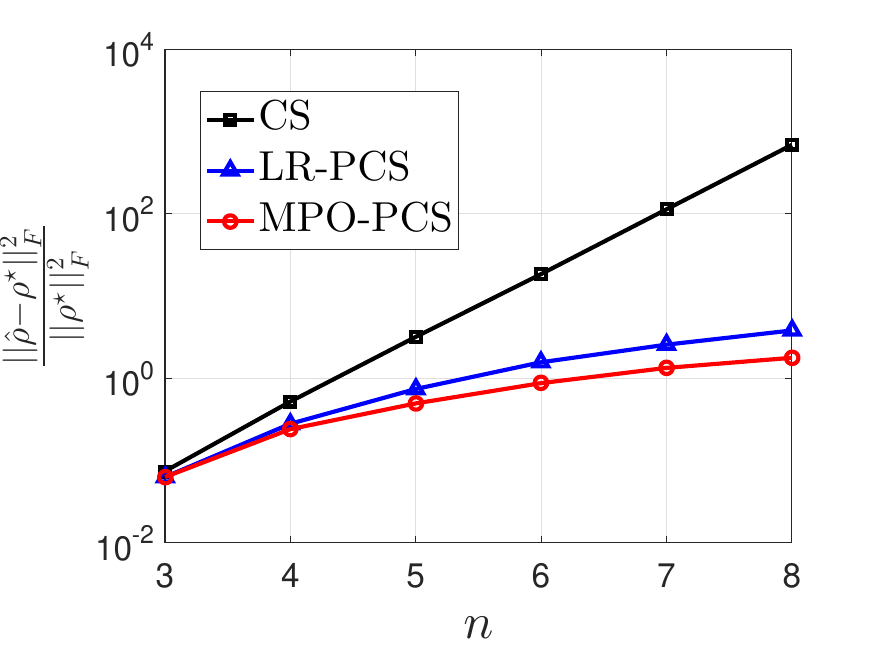}
%\caption{(A)}
\end{minipage}
\label{MSE diff n thermal2}
}
\subfigure[]{
\begin{minipage}[t]{0.31\textwidth}
\centering
\includegraphics[width=5.5cm]{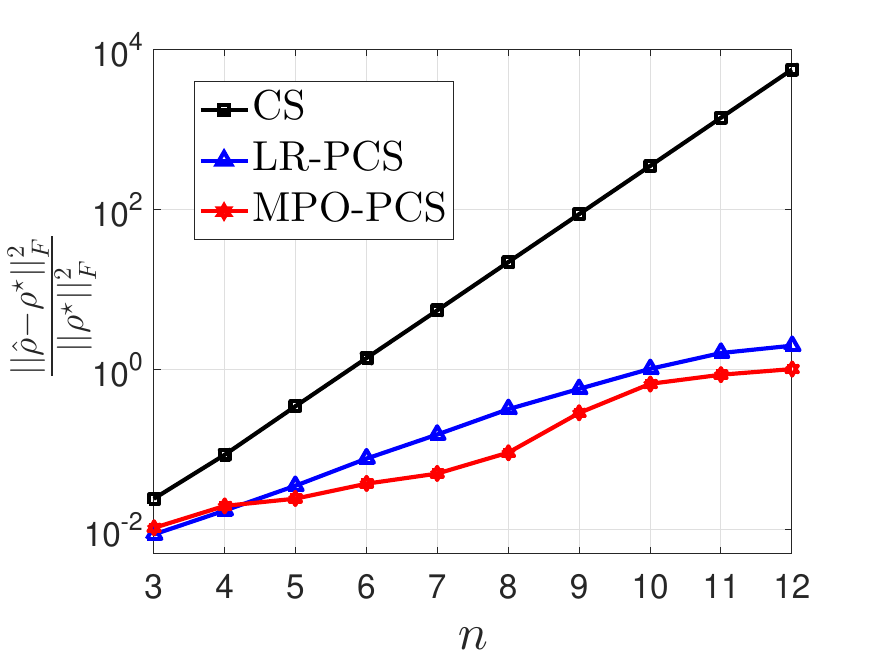}
%\caption{(B)}
\end{minipage}
\label{MSE diff n GHZ}
}
\caption{Mean squared error as a function of the total qubit number with $M = 3000$ for (a) thermal state ($T = 0.2$), (b) thermal state ($T = 2$), and (c) GHZ state.}
\label{F norm for different n}
\end{figure*}

In the second set of trials, we test CS and MPO-PCS across varying numbers of measurements $M$ and bond dimension $D$. We consider $n=7$-qubit matrix product states (MPSs, pure state special cases of MPOs) of the form $\vrho^\star = \vu^\star {\vu^\star }^\dagger\in\C^{128\times 128}$, where $\vu^\star \in \C^{128\times 1}$ satisfies $\|\vu^\star\|_2=1$ and its $(i_1\cdots i_7)$-element can be represented in the matrix product form: $\vu^\star(i_1\cdots i_7) = {\mU_1^\star}^{i_1}\cdots {\mU_7^\star}^{i_7}$. Here, each matrix ${\mU_{\ell}^\star}^{i_{\ell}}$ has size $d \times d$, except for ${\mU_1^\star}^{i_1}$ and ${\mU_7^\star}^{i_7}$ of dimensions of $1 \times d$ and $d \times 1$, respectively.

To generate each MPS $\vu^\star$, we draw a length-128 complex vector with i.i.d. standard normal elements, apply TT-SVD \cite{Oseledets11} to truncate it to an MPS, and then normalize the result to unit length. As a result, entry $\vrho^\star(i_1 \cdots i_7, j_1 \cdots j_7)$ can be expressed as $    \vrho^\star(i_1 \cdots i_7, j_1 \cdots j_7) = ({\mU_1^\star}^{i_1}\otimes {{\mU_1^\star}^{j_1}}^\dagger) \cdots ({\mU_7^\star}^{i_7}\otimes {{\mU_7^\star}^{j_7}}^\dagger) = {\mX_1^\star}^{i_1,j_1}\cdots {\mX_7^\star}^{i_7,j_7}$, where $\otimes$ denotes the Kronecker product. Thus, $\vrho^\star = \vu^\star {\vu^\star }^\dagger$ is also an MPO with bond dimension $D = d^2$ (equal for all qubits). As shown in Fig.~\ref{MPS T2 MSE diff summary} MPO-PCS attains significantly lower error than CS, as it leverages knowledge about the underlying MPO structure. And the recovery error of MPO-PCS increases with higher MPO bond dimension (in line with Table~\ref{Methods comparison MPO}), whereas that of CS remains the same regardless of $D$.

In the third set of trials, we simulate measurements on $7$-qubit density matrices: $(i)$ thermal state
\footnote{The thermal state is generated from the 1D quantum Ising model $H = \sum_{j=1}^{n-1}\sigma_j^z\sigma_{j+1}^z + \sum_{j=1}^{n}\sigma_j^x$ with $\sigma_j^a = \mId_{2^{j-1}} \otimes \sigma^a  \otimes \mId_{2^{n-j}}\in\R^{2^n\times 2^n}, a = x, z$ and $\sigma^x = \begin{bmatrix} 0 & 1 \\ 1 & 0 \end{bmatrix}, \sigma^z = \begin{bmatrix} 1 & 0 \\ 0 & -1 \end{bmatrix}$. The thermal state is then defined as $\vrho^\star = \frac{e^{-H/T}}{\trace(e^{-H/T})}$.}
with temperature $T = 0.2$ (a relatively low temperature close to the ground
state); $(ii)$ thermal state with temperature $T = 2$ (corresponding to a relatively high temperature); and $(iii)$ Greenberger--Horne--Zeilinger (GHZ) state
\footnote{The GHZ state is constructed as $\vrho^\star = \vg \vg^\dagger $ where $\vg =\begin{bmatrix}\frac{1}{\sqrt{2}} & 0 & \cdots & 0 & \frac{1}{\sqrt{2}} \end{bmatrix}^\top\in\R^{2^n\times 1}$.}.
It is worth noting that the low-temperature thermal state $(i)$ and the GHZ state $(iii)$  simultaneously exhibit low-rank and MPO structures \cite{cramer2010efficient,francca2021fast,jameson2024optimal}, making them well-suited for demonstrating the advantages of exploiting structured subspaces.
We impose a rank constraint $r\in\{4,24,1\}$ for the estimator on each state, respectively. For the $T=0.2$ thermal state, the ground-truth density matrix has rank of approximately $4$, while for the high-temperature case ($T=2$), it is full-rank; for LR-PCS $r=24$ is selected, somewhat arbitrarily, which is sufficient to encompass $80\%$ of the sum of the eigenvalues of the ground-truth density matrix. %respectively.
In addition, we apply TT-SVD on the CS estimator to adaptively select the bond dimensions using the error tolerance $10^{-14}$.
\Cref{thermal state and GHZ state MSE} shows that the proposed LR-PCS and MPO-PCS methods outperform standard CS, as quantified by  the Frobenius norm. Additionally, MPO-PCS demonstrates superior performance compared to LR-PCS, which can be attributed to the lower degrees of freedom in the MPO structure relative to the low-rank structure [cf. Eqs.~(\ref{error bound of projected classical shadow F},\ref{error bound of projected classical shadow MPO F})].

In the final test, we examine how the recovery error scales with qubit number $n$, using parameter settings of $r = 4$ for $T = 0.2$, $r = 4(n-1)$ for $T = 2$, $r = 1$ for GHZ state and an error tolerance of $10^{-14}$ for determining bond dimension $D$.
As highlighted in Fig.~\ref{F norm for different n}, both LR-PCS and MPO-PCS effectively attenuate the growth in recovery error as the system size $n$ increases, in contrast to the standard CS method. This improvement is attributed to the utilization of the low-dimensional structure in these methods. Additionally, the recovery error of MPO-PCS scales polynomially with $n$, as indicated in Eq.~\eqref{error bound of projected classical shadow MPO F}, rather than exponentially as in Eq.~\eqref{error bound of projected classical shadow F} of LR-PCS; hence, MPO-PCS outperforms LR-PCS in terms of recovery error.

\section{Conclusion}
\label{conclcusion}

This paper has introduced the projected classical shadow (PCS) method to address the computational challenges of quantum state tomography (QST) in large Hilbert spaces by leveraging the classical shadow (CS) framework combined with a physical projection step. The method provides guaranteed performance under Haar-random measurements. Theoretical results show that the PCS method achieves high accuracy in reconstructing general and low-rank quantum states while minimizing the number of state copies, meeting information-theoretically optimal bounds. Moreover, the PCS method reduces the number of state copies required for matrix product operator (MPO) states compared to existing results using Haar random measurements. Numerical validation further demonstrates the practicality and computational efficiency of PCS for large-scale quantum state reconstruction.

More broadly, our formalism points to a promising new general direction for CS methods. Although originally introduced for the estimation of state \textit{properties} rather than the state \textit{per se}~\cite{huang2020predicting}, CS nevertheless relies on an estimator $\vrho_{\text{CS}}$ of the full density matrix. As our results reveal, this generally unphysical estimator can be projected onto a physical space of interest---whether the entire Hilbert space or some subset thereof (Fig.~\ref{Illustration of projected classical shadow method})---with performance guarantees that  attain information-theoretic bounds (for the case of arbitrary and low-rank states) or improve upon previous scaling results (for MPO states). Therefore in merging the conceptual simplicity of CS with the scaling improvements possible in structured quantum systems, our results suggest a compelling role for PCS in traditional quantum state estimation, with exciting opportunities for future exploration in even more types of subspaces tailored to specific physical conditions or prior knowledge, such as projected entangled pair operator (PEPO) \cite{cirac2021matrix} and multiscale entanglement renormalization
ansatz (MERA) \cite{haegeman2013entanglement}.

\acknowledgments
We acknowledge funding support from the National Science Foundation (CCF-2241298, EECS-2409701) and the U.S. Department of Energy (ERKJ432,  DE-SC0024257). We thank the Ohio Supercomputer Center for providing the computational resources and the Quantum Collaborative led by Arizona State University for providing valuable expertise and resources. A portion of this work was performed at Oak Ridge National Laboratory, operated by UT-Battelle for the U.S. Department of Energy under Contract No. DE-AC05-00OR22725.

%\bibliography{reference1}% Produces the bibliography via BibTeX.


\appendix
\onecolumngrid

\section{Proof of Equation~(\ref{proof of expectation of MSE main result})}
\label{proof of expectation of MSE}

\begin{proof}
    We expand $\E\|\vrho_{\text{CS}}  - \vrho^\star \|_F^2$ as follows:
\begin{eqnarray}
\label{proof of expectation of MSE1}
    \E\|\vrho_{\text{CS}}  - \vrho^\star \|_F^2&=& \E\left\| \frac{1}{M}\sum_{m=1}^M\vrho_m - \vrho^\star\right\|_F^2\nonumber\\
     & =& \E\left\<\frac{1}{M}\sum_{m=1}^M (\vrho_m - \vrho^\star), \frac{1}{M}\sum_{m=1}^M (\vrho_m - \vrho^\star)\right\>\nonumber\\
     & =& \frac{1}{M^2}\E\sum_{m=1}^M\|\vrho_m  - \vrho^\star \|_F^2\nonumber\\
     & =&  \frac{1}{M}\E\|\vrho_1  - \vrho^\star \|_F^2\nonumber\\
     & =& \frac{1}{M}(\|\vrho^\star\|_F^2 - 2\E\<\vrho_1, \vrho^\star   \> + \E\<\vrho_1,\vrho_1  \>)\nonumber\\
     & =& \frac{1}{M}\bigg[- \|\vrho^\star \|_F^2 + (2^n+1)^2\E\<\vphi_{1,j_1}\vphi_{1,j_1}^\dagger, \vphi_{1,j_1}\vphi_{1,j_1}^\dagger  \>-2(2^n+1)\E\<\vphi_{1,j_1}\vphi_{1,j_1}^\dagger, \mId_{2^n}  \> +2^n \bigg]\nonumber\\
     & =&\frac{4^n + 2^n - 1 - \|\vrho^\star \|_F^2}{M}
\end{eqnarray}
where the third line follows from $\E[\vrho_m] = \vrho^\star$, the fourth from the equivalence under expectation of all measurements $m$, and the last from the normalization $\<\vphi_{1,j_1}\vphi_{1,j_1}^\dagger, \vphi_{1,j_1}\vphi_{1,j_1}^\dagger  \> = \<\vphi_{1,j_1}\vphi_{1,j_1}^\dagger, \mId_{2^n}  \> = 1$.

\end{proof}

\section{Proof of {Theorem} \ref{project CS:general}}
\label{proof of error bound for general state}

\begin{proof}
We define a restricted Frobenius norm as
\begin{eqnarray}
    \label{Definition of the restricted F norm1 CS general state}
    \|\vrho_{\text{PCS}} - \vrho^\star\|_{F,\wh\setX} &=& \|\vrho_{\text{PCS}} - \vrho^\star\|_{F}= \max_{\vrho\in \wh\setX} \<\vrho_{\text{PCS}} - \vrho^\star,  \vrho  \>.
\end{eqnarray}
Note that $\wh\setX$, includes conditions such as $\trace(\vrho) = 0$, $\vrho = \vrho^\dagger$ and $\|\vrho\|_F\leq 1$. By the definition of the restricted  Frobenius norm in Eq.~\eqref{Definition of the restricted F norm1 CS general state}, we can further analyze
\begin{equation}
    \label{Expansion of  projected CS general state}
    \|\vrho_{\text{PCS}} - \vrho^\star \|_F
     = \|\vrho_{\text{PCS}} - \vrho^\star \|_{F,\wh\setX}
     \leq  \|\vrho_{\text{CS}}  - \vrho^\star \|_{F,\wh\setX}
    =  \max_{\vrho\in \wh\setX } \bigg\<\frac{1}{M}\sum_{m=1}^M \left[(2^n+1) \vphi_{m,j_m}\vphi_{m,j_m}^\dagger - \mId_{2^n}\right] - \vrho^\star, \vrho  \bigg\>,
\end{equation}
where the inequality follows from the assumption that the physical projection $\calP_{\setX}(\cdot)$ is optimal and therefore satisfies nonexpansiveness.
Next, we bound $\frac{1}{M}\sum_{m=1}^M [(2^n+1) \vphi_{m,j_m}\vphi_{m,j_m}^\dagger - \mId_{2^n}] - \vrho^\star$ using the covering argument. According to the assumption, we initially construct an $\epsilon$-net $\{ \vrho^{(1)}, \dots, \vrho^{(N_{\epsilon}(\wt \setX))} \}\in \wt \setX \subset \wh \setX$, where the size of $\wt \setX$ is denoted by $N_{\epsilon}(\wt \setX)$ such that
\begin{eqnarray}
    \label{set of covering number general state}
    \sup_{\vrho: \|\vrho\|_F \leq 1}\min_{p\leq N_{\epsilon}(\wt \setX)} \|\vrho - \vrho^{(p)}\|_F\leq \epsilon.
\end{eqnarray}
In addition, we denote $\mB_m = \frac{1}{M}((2^n+1) \vphi_{m,j_m}\vphi_{m,j_m}^\dagger - \mId_{2^n}- \vrho^\star)$ and derive
\begin{equation}
    \label{relationship between fix and any}
    \max_{\vrho\in \wh \setX}  \bigg\< \sum_{m=1}^M \mB_m, \vrho \bigg\>
    =\max_{\vrho\in \wh \setX}  \bigg\< \sum_{m=1}^M \mB_m, \vrho - \vrho^{(p)} + \vrho^{(p)} \bigg\>\nonumber
    \leq  \max_{\vrho^{(p)}\in  \wt\setX}\bigg\< \sum_{m=1}^M \mB_m, \vrho^{(p)} \bigg\>  + \epsilon \max_{\vrho\in \wh \setX}  \bigg\< \sum_{m=1}^M \mB_m, \vrho \bigg\>.
\end{equation}

By setting $\epsilon = 0.5$ and moving the second term on the right-hand side to the left, we get
\begin{eqnarray}
    \label{relationship between fix and any 1}
    \max_{\vrho\in \wh \setX}  \bigg\< \sum_{m=1}^M \mB_m, \vrho \bigg\> \leq \max_{\vrho^{(p)}\in \wt \setX}  2\bigg\< \sum_{m=1}^M \mB_m, \vrho^{(p)} \bigg\>.
\end{eqnarray}
Then we need to build the concentration inequality for the right hand side of Eq.~\eqref{relationship between fix and any 1}. First, we define
\begin{eqnarray}
    \label{definition of sum in concentration inequality}
    \sum_{m = 1}^{M} s_m= \sum_{m=1}^M\<(2^n+1) \vphi_{m,j_m}\vphi_{m,j_m}^\dagger - \mId_{2^n} - \vrho^\star, \vrho^{(p)}  \>,\nonumber\\
\end{eqnarray}
and due to $\E[ (2^n+1) \vphi_{m,j_m}\vphi_{m,j_m}^\dagger - \mId_{2^n} - \vrho^\star] = {\bf 0}$, we have $\E[s_m] = 0$. Moreover, we rewrite $s_m$ as
\begin{eqnarray}
    \label{change of variable sk}
    s_m&=& \<(2^n+1) \vphi_{m,j_m}\vphi_{m,j_m}^\dagger - \mId_{2^n} - \vrho^\star, \vrho^{(p)}  \>\nonumber\\
    &=& (2^n+1) \bigg\< \vphi_{m,j_m}\vphi_{m,j_m}^\dagger - \frac{\vrho^\star}{2^n+1}, \vrho^{(p)}    \bigg\>\nonumber\\
    &=& (2^n+1) \bigg\< \vphi_{m,j_m}\vphi_{m,j_m}^\dagger,  \vrho^{(p)} - \frac{\<\vrho^\star, \vrho^{(p)}  \>}{2^n+1} \mId_{2^n}  \bigg \>\nonumber\\
    &=& (2^n+1) \< \vphi_{m,j_m}\vphi_{m,j_m}^\dagger,  \mD   \>,
\end{eqnarray}
where the second line follows from  $\trace(\vrho^{(p)}) = \<\mId_{2^n}, \vrho^{(p)} \> = 0$. We can further compute
\begin{eqnarray}
    \label{p moment for change of variable sk}
   \E[|s_m|^a]  &=& \E[(2^n+1)^a |\< \vphi_{m,j_m}\vphi_{m,j_m}^\dagger,  \mD   \> |^a ]\nonumber\\
   &\leq& (2^n+1)^a\E[ (\trace(\vphi_{m,j_m}\vphi_{m,j_m}^\dagger |\mD | ))^a  ]\nonumber\\
   &=& \frac{(2^n+1)^a}{C_{2^n+a-1}^a} \trace(|\mD |^{\otimes a} P_{\text{Sym}})\nonumber\\
   &\leq& \frac{(2^n+1)^a}{C_{2^n+a-1}^a} \||\mD | \|_F^{\otimes a} \|P_{\text{Sym}}\| \nonumber\\
   &\leq& 6\times 2^{a-2} a!,
\end{eqnarray}
where $|\mD| = \sqrt{\mD^2} = \mU\sqrt{\mSigma}\mV^\dagger$ denotes the absolute value of the matrix $\mD$ with its compact SVD $\mD^2 = \mU\mSigma\mV^\dagger$ and $\mA^{\otimes a} = \underbrace{\mA\otimes \cdots \otimes \mA}_a$ holds for any matrix $\mA$.
Given that the unitary Haar measure conforms to any unitary $p$-design, as exemplified in Ref.~\cite[Example 51]{mele2023introduction}, we can deduce the third line, with $P_{\text{Sym}}$ representing an orthogonal projector onto the symmetric subspace. The second inequality follows from \cite[Lemma 7]{zhu2021global} and $\||\mD | \|_F^{\otimes a} = \||\mD |^{\otimes a} \|_F$ due to the positive semidefiniteness of $|\mD |^{\otimes a}$ and the orthogonal projection. In the last line, we utilize $\||\mD | \|_F \leq \|\vrho^{(p)} \|_F + \|\frac{\<\vrho^\star, \vrho^{(p)}  \>}{2^n+1} \mId_{2^n}\|_F\leq 1  + \frac{2^n }{2^n+1}\|\vrho^{(p)}\|_F\|\vrho^\star\|_F\leq 2$, $\|P_{\text{Sym}}\|\leq 1$ and $\frac{(2^n+1)^a}{C_{2^n+a-1}^a} \leq \frac{3}{2}a!$.

Based on \Cref{Classical Bernstein inequality} with $\E[s_m] = 0$ and $\E[|s_m|^a]\leq 6\times 2^{a-2} a!$, for any $t\in[0,1]$, we have the probability

\begin{equation}
    \label{concentration inequality for fixed absolute term}
    \P{ \frac{1}{M}\bigg|\sum_{m=1}^M\bigg\< (2^n+1) \vphi_{m,j_m}\vphi_{m,j_m}^\dagger - \mId_{2^n} - \vrho^\star, \vrho^{(p)}  \bigg\>\bigg|\geq t } \leq 2 e^{-\frac{Mt^2}{28}}.
\end{equation}

Combining Eqs.~(\ref{relationship between fix and any 1},\ref{concentration inequality for fixed absolute term}),  there
exists an $\epsilon$-net $\wt \setX$ of $\wh \setX$  such that
{\small\begin{eqnarray}
    \label{concentration inequality for random term general state}
    \P{\max_{\vrho\in \wh \setX} \bigg\< \frac{1}{M}\sum_{m=1}^M \left[(2^n+1) \vphi_{m,j_m}\vphi_{m,j_m}^\dagger \!\!-\!\! \mId_{2^n}\right] \!\!-\!\! \vrho^\star, \vrho  \bigg\> \geq t}
    &\leq&\P{\!\max_{\vrho^{(p)}\in \wt \setX} \frac{1}{M}\sum_{m=1}^M\big\< (2^n\!\!+\!\!1) \vphi_{m,j_m}\vphi_{m,j_m}^\dagger \!\!-\!\! \mId_{2^n} \!\!-\!\! \vrho^\star, \vrho^{(p)}  \big\>\geq \frac{t}{2} }\nonumber\\
    &\leq&\P{\!\max_{\vrho^{(p)}\in \wt \setX} \frac{1}{M}\bigg|\sum_{m=1}^M\big\< (2^n\!\!+\!\!1) \vphi_{m,j_m}\vphi_{m,j_m}^\dagger \!\!-\!\! \mId_{2^n} \!\!-\!\! \vrho^\star, \vrho^{(p)}  \big\>\bigg|\geq \frac{t}{2} }\nonumber\\
    &\leq&  2N_{\epsilon}(\wt \setX) e^{-\frac{Mt^2}{112}}\nonumber\\
    &\leq& e^{-\frac{Mt^2}{112} + \log 2N_{\epsilon}(\wt \setX)}.
\end{eqnarray}}
We opt for $t = O\left(\sqrt{\frac{\log N_{\epsilon}(\wt \setX)}{M}}\right)$, and subsequently, with probability $1- e^{- \Omega(\log N_{\epsilon}(\wt \setX))}$, we derive
\begin{eqnarray}
    \label{upper bound F norm of projected general state}
    \|\vrho_{\text{PCS}} - \vrho^\star\|_F \leq O\left(\sqrt{\frac{\log N_{\epsilon}(\wt \setX)}{M}}\right).
\end{eqnarray}

\end{proof}

\section{Proof of {Theorem} \ref{project CS conclusion}}
\label{Proof of error bound in projected CS}

\begin{proof}

We define a restricted Frobenius norm as following:
\begin{eqnarray}
    \label{Definition of the restricted F norm1 CS}
    &&\|\vrho_1 - \vrho_2\|_{F,2r} = \|\vrho_1 - \vrho_2\|_F = \max_{\vrho\in\wh \setX_{2r}} \<\vrho_1 - \vrho_2,  \vrho  \>,
\end{eqnarray}
where the set $\wh\setX_{r}$ is defined as follows:
\begin{eqnarray}
    \label{The set of normalized rho CS}
    \wh \setX_{r} = \{\vrho\in\C^{2^n\times 2^n}:  \vrho = \vrho^\dagger, \text{rank}(\vrho)= r,  \trace(\vrho) = 0, \|\vrho\|_F \leq 1 \}.
\end{eqnarray}

By the definition of the restricted  Frobenius norm in Eq.~\eqref{Definition of the restricted F norm1 CS}, we can further analyze
\begin{eqnarray}
    \label{Expansion of  projected CS}
    \|\vrho_{\text{LR-PCS}} - \vrho^\star \|_F\nonumber
    &=& \|\vrho_{\text{LR-PCS}} - \vrho^\star \|_{F,2r}\nonumber\\
    &\leq& \|\calP_{\text{ED}}(\vrho_{\text{CS}})  - \vrho^\star \|_{F,2r}\nonumber\\
    &\leq& 2\|\vrho_{\text{CS}} - \vrho^\star \|_{F,2r}\nonumber\\
    &=& 2\max_{\vrho\in \wh \setX_{2r}} \bigg\<\frac{1}{M}\sum_{m=1}^M [(2^n+1) \vphi_{m,j_m}\vphi_{m,j_m}^\dagger - \mId_{2^n}] \!\!-\!\! \vrho^\star, \vrho  \bigg\>,\nonumber\\
\end{eqnarray}
where the first two inequalities respectively follow the nonexpansiveness property of the projection and the quasi-optimality property of eigenvalue decomposition (ED) projection \cite{Oseledets11}. Next, we need bound the first term in the last line of Eq.~\eqref{Expansion of  projected CS} using the covering argument. According to Ref.~\cite[Lemma 3.1]{Candes11}, we initially construct an $\epsilon$-net $\{ \vrho^{(1)}, \dots, \vrho^{N_{\epsilon}(\wt \setX_{2r})} \}\in \wt \setX_{2r} \subset \wh \setX_{2r}$  in which the size of $\wt \setX_{2r}$ is denoted by $N_{\epsilon}(\wt \setX_{2r})\leq (\frac{9}{\epsilon})^{(2^{n+2} + 2 )r}$ such that
\begin{eqnarray}
    \label{set of covering number}
    \sup_{\vrho: \|\vrho\|_F\leq 1}\min_{p\leq N_{\epsilon}(\wt \setX_{2r})} \|\vrho - \vrho^{(p)}\|_F\leq \epsilon.
\end{eqnarray}

Combining Eqs.~(\ref{relationship between fix and any 1},\ref{concentration inequality for fixed absolute term}) in {Appendix}~\ref{proof of error bound for general state},  there
exists an $\epsilon$-net $\wt \setX_{2r}$ of $\wh \setX_{2r}$  such that
{\small \begin{eqnarray}
    \label{concentration inequality for random term low-rank}
    \P{\max_{\vrho\in \wh \setX_{2r}}\<\frac{1}{M}\sum_{m=1}^M \left[(2^n\!\!+\!\!1) \vphi_{m,j_m}\vphi_{m,j_m}^\dagger \!\!-\!\! \mId_{2^n}\right] \!\!-\!\! \vrho^\star, \vrho  \> \geq t}
    &\leq&\P{\!\max_{\vrho^{(p)}\in \wt \setX_{2r}} \frac{1}{M}\sum_{m=1}^M\< (2^n\!\!+\!\!1) \vphi_{m,j_m}\vphi_{m,j_m}^\dagger \!\!-\!\! \mId_{2^n} \!\!-\!\! \vrho^\star, \vrho^{(p)}  \>\geq \frac{t}{2} }\nonumber\\
    &\leq&\P{\!\max_{\vrho^{(p)}\in \wt \setX_{2r}} \frac{1}{M}\bigg|\sum_{m=1}^M\< (2^n\!\!+\!\!1) \vphi_{m,j_m}\vphi_{m,j_m}^\dagger \!\!-\!\! \mId_{2^n} \!\!-\!\! \vrho^\star, \vrho^{(p)}  \>\bigg|\geq \frac{t}{2} }\nonumber\\
    &\leq&  2 \bigg(\frac{9}{\epsilon}\bigg)^{(2^{n+2} + 2 )r} e^{-\frac{Mt^2}{112}}\nonumber\\
    &\leq& e^{-\frac{Mt^2}{112} + C 2^n r},
\end{eqnarray}}
where we set $\epsilon=\frac{1}{2}$ and $C$ is a positive constant. We opt for $t = O\left(\sqrt{\frac{2^n r}{M}}\right)$ and subsequently, with probability $1- e^{- \Omega(2^n r)}$,  derive
\begin{eqnarray}
    \label{upper bound F norm of projected }
    \|\vrho_{\text{LR-PCS}} - \vrho^\star\|_F \leq O\bigg(\sqrt{\frac{2^n r}{M}}\bigg).
\end{eqnarray}

\end{proof}

\section{Proof of {Theorem} \ref{project CS MPO conclusion}}
\label{Proof of project CS MPO conclusion}

We define a restricted Frobenius norm as following:
\begin{eqnarray}
    \label{Definition of the restricted F norm1}
    &&\|\vrho_1 - \vrho_2\|_{F,2D} = \|\vrho_1 - \vrho_2\|_F = \max_{\vrho\in\wh \setX_{2D}} \<\vrho_1 - \vrho_2,  \vrho  \>.
\end{eqnarray}
where we denote by $\wh \setX_{D}$ the normalized set of MPOs with bond dimension $D$:
\begin{eqnarray}
\label{SetOfMPO normalized}
\wh \setX_{D}= \Big\{ \vrho\in\C^{2^n\times 2^n}:\ \vrho = \vrho^\dagger, \|\vrho\|_F \leq 1, \trace(\vrho) = 0,  \text{bond dimension}(\vrho) = D \Big\}.
\end{eqnarray}
Note that the presence of additional orthonormal structures arises from the fact that, according to Ref.~\cite{holtz2012manifolds}, any TT form is equivalent to a left-orthogonal TT form \cite{Oseledets11}.

We define  $\calP_{\trace}(\cdot)$ as a projection onto convex set $\{\vrho\in\C^{2^n\times 2^n}: \trace(\vrho) = 1 \}$.
By the definition of the restricted  Frobenius norm  \eqref{Definition of the restricted F norm1}, we can derive
\begin{eqnarray}
\label{proof of upper bound spec}
    \|\vrho_{\text{MPO-PCS}} - \vrho^\star\|_F    &\leq & \|\calP_{\trace}(\text{SVD}_{D}^{tt}(\vrho_{\text{CS}})) - \vrho^\star \|_{F}\nonumber\\
    &=& \|\calP_{\trace}(\text{SVD}_{D}^{tt}(\vrho_{\text{CS}})) - \vrho^\star  \|_{F, 2D}\nonumber\\
    &\!\!\!\!\leq\!\!\!\!&  \| \text{SVD}_{D}^{tt}(\vrho_{\text{CS}}) - \vrho^\star\|_{F, 2D}\nonumber\\
    &\!\!\!\!\leq\!\!\!\!& (1+ \sqrt{n-1}) \| \vrho_{\text{CS}} - \vrho^\star  \|_{F, 2D}\nonumber\\
    &=& (1+ \sqrt{n-1})\max_{\vrho\in \wh \setX_{2D}} \bigg\<\frac{1}{M}\sum_{m=1}^M ((2^n+1) \vphi_{m,j_m}\vphi_{m,j_m}^\dagger- \mId_{2^n}) - \vrho^\star, \vrho  \bigg\>\nonumber\\
    &=&(1+ \sqrt{n-1})\max_{\vrho\in \wh \setX_{2D}} \bigg\<\frac{1}{M}\sum_{m=1}^M ((2^n+1) \vphi_{m,j_m}\vphi_{m,j_m}^\dagger- \mId_{2^n}) - \vrho^\star, \vrho  \bigg\>
\end{eqnarray}
where the first two inequalities respectively follow from the nonexpansiveness property of the projection onto the convex set, while the third inequality is a consequence of the quasi-optimality property of TT-SVD projection \cite{Oseledets11}. Additionally, we denote
\begin{eqnarray}
\label{SetOfMPO normalized another}
\wh \setX_{D}= \Big\{&& \vrho\in\C^{2^n\times 2^n}:\ \vrho = \vrho^\dagger, \trace(\vrho) = 0, \vrho(i_1 \cdots i_\nqbit, j_1\cdots j_\nqbit) = \mX_1^{i_1,j_1} \mX_2^{i_2,j_2} \cdots \mX_\nqbit^{i_\nqbit,j_\nqbit}, \nonumber\\
&& \mX_1^{i_1,j_1}\in\C^{1\times D},\mX_n^{i_n,j_n}\in\C^{D\times 1},\mX_\ell^{i_\ell,j_\ell}\in\C^{D\times D},  \|L(\mX_\ell) \|\leq 1,\ell\in[n-1], \|L(\mX_n)\|_F\leq 1 \Big\}.
\end{eqnarray}
Based on $\|\vrho\|_F=  \|L(\mX_n)\|_F \leq 1$ for a left-orthogonal TT form using \cite[Eq.(44)]{qin2024guaranteed}, we obtain the last line.

Next, we will apply the covering argument to bound \eqref{proof of upper bound spec}. For any fixed value of $\widetilde\vrho\in \wt \setX_{2D} \subset \wh \setX_{2D}$, using Eq.~\eqref{relationship between fix and any 1}, concentration inequality in Eq.~\eqref{concentration inequality for fixed absolute term} and \Cref{lem:cover-number-left-orthogonal MPO}, there exists an $\epsilon$-net $\wt \setX_{2D}$ of $\wh \setX_{2D}$  such that
{\small \begin{eqnarray}
    \label{concentration inequality for random term}
    \P{\max_{ \vrho\in \wh \setX_{2D}}\<\frac{1}{M}\sum_{m=1}^M ((2^n+1) \vphi_{m,j_m}\vphi_{m,j_m}^\dagger \!\!-\!\! \mId_{2^n}) \! - \!\!\vrho^\star, \vrho  \> \geq t}
    &\leq&\P{\max_{ \wt\vrho\in \wt \setX_{2D}} \frac{1}{M}\sum_{m=1}^M\< (2^n+1) \vphi_{m,j_m}\vphi_{m,j_m}^\dagger \!\!- \!\!\mId_{2^n} \!\!-\!\! \vrho^\star, \wt\vrho  \>\geq \frac{t}{2} }\nonumber\\
    &\leq&\P{\max_{ \wt\vrho\in \wt \setX_{2D}} \frac{1}{M}\bigg|\sum_{m=1}^M\< (2^n+1) \vphi_{m,j_m}\vphi_{m,j_m}^\dagger \!\!-\!\! \mId_{2^n} \!\!- \!\! \vrho^\star, \wt\vrho  \>\bigg|\geq \frac{t}{2} }\nonumber\\
    &\leq&  2\bigg(\frac{4n+\epsilon}{\epsilon}\bigg)^{4n D^2} e^{-\frac{Mt^2}{112}}\nonumber\\
    &\leq& e^{-\frac{Mt^2}{112} + C nD^2 \log n},
\end{eqnarray}}
where we set $\epsilon=\frac{1}{2}$ and $C$ is a positive constant. We opt for $t = O\left(\sqrt{\frac{nD^2 \log n}{M}}\right)$ and subsequently, with probability $1- e^{- \Omega(nD^2 \log n)}$, derive
\begin{eqnarray}
    \label{proof of upper bound spec final}
    \|\vrho_{\text{MPO-PCS}} - \vrho^\star\|_F \leq O\bigg(\sqrt{\frac{n^2D^2 \log n}{M}}\bigg).
\end{eqnarray}

\section{Auxiliary Materials}
\label{sec: Auxiliary Materials}
\begin{lemma}
\label{Classical Bernstein inequality}(Classical Bernstein's inequality \cite[Theorem 6]{guctua2020fast})
Let $s_1,\dots, s_n\in\R$ denote i.i.d. copies of a mean-zero random variable $s$ that obeys $\E[|s|^p]\leq p!R^{p-2} \sigma^2/2$ for all integers $p\geq 2$, where $R,\sigma^2 > 0$  are constants. Then, for all $t > 0$,
\begin{eqnarray}
\label{classical Bernstein}
\P{\bigg|\sum_{i=1}^n s_i  \bigg|\geq t } \leq 2 e^{-\frac{t^2/2}{n\sigma^2 + Rt}}.
\end{eqnarray}
\end{lemma}

\begin{lemma}(\cite[Lemma 10]{qin2024quantum})
\label{EXPANSION_A1TOAN-B1TOBN_1}
For any  ${\bm A}_i,{\bm A}^\star_i\in\R^{r_{i-1}\times r_i},i\in\{1,\dots,N\}$, we have
\begin{eqnarray}
    \label{EXPANSION_A1TOAN-B1TOBN_2}
    &&{\bm A}_1{\bm A}_2\cdots {\bm A}_N-{\bm A}_1^\star{\bm A}_2^\star\cdots {\bm A}_N^\star= \sum_{i=1}^N \mA_1^\star \cdots \mA_{i-1}^\star (\mA_{i} - \mA_i^\star) \mA_{i+1} \cdots \mA_N.
\end{eqnarray}
\end{lemma}

\begin{lemma}
There exists an $\epsilon$-net $\widetilde\setX_{D}$ for $\wh\setX_{D}$ in Eq.~\eqref{SetOfMPO normalized another} under the Frobenius norm, i.e., $\|\vrho - \vrho^{(p)}\|_F\leq \epsilon$ for $\vrho^{(p)}\in \widetilde\setX_{D}$, obeying
\begin{eqnarray}
    \label{lem:cover-number-normalized left-orthogonal MPO}
    N_{\epsilon}(\widetilde\setX_{D}) \le \bigg(\frac{4n+\epsilon}{\epsilon}\bigg)^{4nD^2},
\end{eqnarray}
where $N_{\epsilon}(\widetilde\setX_{D})$ denotes the number of elements in the set $\widetilde\setX_{D}$.
\label{lem:cover-number-left-orthogonal MPO}\end{lemma}
\begin{proof}
For each set of matrices $\{L(\mX_\ell)\in\R^{4D\times D}: \|L(\mX_\ell)\|\leq 1\}$, according to Ref.~\cite{zhang2018tensor}, we can construct an $\xi$-net $\{L(\mX_\ell^{(1)}), \dots, L(\mX_\ell^{(N_\ell)})  \}$ with the covering number $N_\ell\leq (\frac{4+\xi}{\xi})^{4D^2}$ such that
\begin{eqnarray}
    \label{ProofOf<H,X>forSubgaussian_proof1}
    \sup_{L(\mX_\ell): \|L(\mX_\ell)\|\leq 1}~\min_{p_\ell\leq N_\ell} \|L(\mX_\ell)-L(\mX_\ell^{(p_\ell)})\|\leq \xi,
\end{eqnarray}
for all $\ell\in\{1,\dots, n-1\}$.
Also, we can construct an $\xi$-net $\{ L(\mX_n^{(1)}), \dots, L(\mX_n^{(N_n)}) \}$ for $\{L(\mX_n)\in\R^{4D\times 1}: \|L(\mX_n)\|_F\leq 1  \}$ such that
\begin{eqnarray}
    \label{ProofOf<H,X>forSubgaussian_proof2}
    \hspace{-0.5cm}\sup_{L(\mX_n): \|L(\mX_n)\|_F\leq 1}\min_{p_n\leq N_n} \|L(\mX_n)-L(\mX_n^{(p_n)})\|_F\leq \xi,
\end{eqnarray}
with the covering number $N_n\leq (\frac{2+\xi}{\xi})^{4D}$.

Therefore, we can construct a $\xi$-net $\{[\mX_1^{(1)}, \dots,  \mX_n^{(1)} ], \ldots, [\mX_1^{(N_1)}, \dots, \mX_n^{(N_n)} ]\}$ with covering number
\begin{eqnarray}
    \label{total covering number}
    \Pi_{\ell=1}^n N_\ell \leq \bigg(\frac{4+\xi}{\xi}\bigg)^{4nD^2}
\end{eqnarray}
for any MPO $\vrho = [\mX_1,\dots, \mX_n]$ with bond dimension $D$.
Then we expand $\|\vrho - \vrho^{(p)} \|_F$ as follows:
{\small \begin{eqnarray*}
    \label{expansion of difference in the covering number}
    \|\vrho - \vrho^{(p)} \|_F    &=& \|  [\mX_1,\dots, \mX_n] - [\mX_1^{(p_1)},\dots, \mX_n^{(p_n)}] \|_F\nonumber\\
    &=& \| \sum_{a_l=1}^n[\mX_1^{(p_1)},\dots, \mX_{a_l-1}^{(p_l)}, \mX_{a_l}^{(p_{a_l})}\!\!-\!\!\mX_{a_l}, \mX_{a_l+1},  \dots, \mX_n]  \|_F\nonumber\\
    &\leq &  \sum_{a_l=1}^n \|[\mX_1^{(p_1)},\dots, \mX_{a_l-1}^{(p_l)}, \mX_{a_l}^{(p_{a_l})}\!\!-\!\!\mX_{a_l}, \mX_{a_l+1},  \dots, \mX_n]  \|_F\nonumber\\
    &\leq & \sum_{a_l=1}^{n-1} \|L(\mX_{a_l}^{(p_{a_l})})\!-\! L(\mX_{a_l}) \| + \| L(\mX_{n}^{(p_{n})}) \!-\! L(\mX_{n}) \|_F  \nonumber\\
    &\leq & n \xi = \epsilon,
\end{eqnarray*}}
where the second line and the second inequality respectively follow \Cref{EXPANSION_A1TOAN-B1TOBN_1} and \cite[Eq.(47)]{qin2024quantum}. In addition, we choose $\xi = \frac{\epsilon}{n}$ in the last line. Ultimately, we can construct an $\epsilon$-net $\{\vrho^{(1)},\ldots, \vrho^{N_1\cdots N_n}\}$ with covering number
\begin{eqnarray}
    \label{total covering number epsilon conclusion}
    N_{\epsilon}(\widetilde\setX_{D}) \le \bigg(\frac{4n+\epsilon}{\epsilon}\bigg)^{4nD^2}
\end{eqnarray}
for any MPO $\vrho\in \wh\setX_{D}$.

\end{proof}

\end{document}